%% file: main.tex
\pdfoutput=1
\documentclass{LMCS}

\def\doi{8(4:10)2012}
\lmcsheading%
{\doi}
{1--24}
{}
{}
{Jan.~\phantom04, 2010}
{Oct.~23, 2012}
{}

\usepackage{hyperref}
\usepackage{enumerate}

\usepackage{amsmath}
\usepackage{amssymb}
\usepackage{algorithmic}
\usepackage{algorithm}

\usepackage{prettyref}
\newrefformat{def}{Definition \ref{#1}}
\newrefformat{cor}{Corollary \ref{#1}}
\newrefformat{prop}{Proposition \ref{#1}}
\newrefformat{alg}{Algorithm \ref{#1}}
\newrefformat{fig}{Figure \ref{#1}}

\usepackage{tikz}
\usetikzlibrary{positioning}
\usetikzlibrary{automata}
\usetikzlibrary{calc}
\usetikzlibrary{fadings}
\usetikzlibrary{fit}
\usetikzlibrary{shapes.misc}
\usetikzlibrary{shapes.geometric}
\usetikzlibrary{backgrounds}

\theoremstyle{plain}\newtheorem{thm2}{Theorem}

\newcommand{\comment}[1]{}

\input{macros}

\begin{document}

\title[Efficient Parallel Path Checking for LTL With Past and Bounds]{Efficient Parallel Path Checking for Linear-Time Temporal Logic With Past and Bounds\rsuper*}

\author[L.~Kuhtz]{Lars Kuhtz}
\address{Universit\"at des Saarlandes\\ 66123 Saarbr\"ucken, Germany}
\email{\{kuhtz, finkbeiner\}@cs.uni-saarland.de}
\thanks{This work was partly supported by the German Research Foundation (DFG) as part of the Transregional Collaborative Research Center ``Automatic Verification and Analysis of Complex Systems'' (SFB/TR 14 AVACS)}

\titlecomment{{\lsuper*}\prettyref{thm:mainltl} of this paper has been published before in \cite{KuhtzFinkbeiner09}. The results in this paper are also included in the Ph.D.\ thesis of the first author \cite{Kuhtz10}.}
 
\author{Bernd Finkbeiner}
\address{\vskip-6 pt}










\begin{abstract}
  Path checking, the special case of the model checking problem where
  the model under consideration is a single path, plays an important
  role in monitoring, testing, and verification.  We prove that for
  linear-time temporal logic (\LTL), path checking can be efficiently
  parallelized.  In addition to the core logic, we consider the
  extensions of \LTL\ with bounded-future (BLTL) and past-time (\PLTL) operators.
  Even though both extensions improve the succinctness of the logic
  exponentially, path checking remains efficiently parallelizable: Our
  algorithm for \LTL, \LTL+Past, and \BLTL\
  is in $\ACO(\logDCFL) \subseteq \NC$.
 \end{abstract}

\maketitle

\section{Introduction}






\noindent Linear-time temporal logic (\LTL) is the standard specification
language to describe properties of computation paths.  The problem of
checking whether a given finite path satisfies an \LTL\ formula plays a key role
in monitoring and runtime
verification~\cite{Havelund+Rosu/04/Efficient,Finkbeiner+Sipma/04/Checking,Dahan+others/05/Combining,ArmoniKorchemnyTiemeyerVardiZbar06,BouleZilic2008},
where individual paths are checked either online, during the execution of the
system, or offline, for example based on an error report. Similarly, path
checking occurs in testing~\cite{Artho+others/05/Combining} and in several
static verification techniques, notably in Monte-Carlo-based probabilistic
verification, where large numbers of randomly generated sample paths are
analyzed~\cite{Younes+Simmons/02/Probabilistic}. 

Somewhat surprisingly, given the widespread use of \LTL, the complexity of the
path checking problem is still open~\cite{Markey+Schnoebelen/03/Model}.  The
established upper bound is $\poly$: The algorithms in the literature traverse
the path sequentially
(cf.~\cite{Finkbeiner+Sipma/04/Checking,Markey+Schnoebelen/03/Model,Havelund+Rosu/04/Efficient});
by going backwards from the end of the path, one can ensure that, in each step,
the value of each subformula is updated in constant time, which results in
running time that is quadratic in the size of the formula plus the length
of the path.
\comment{Well, it is easy to see, that for either parameter fixed the complexity
is in \logspace.}
The only known lower bound is
$\NCO$~\cite{Demri+Schnoebelen/02/Complexity}, the complexity of evaluating
Boolean expressions. 
The large gap between the bounds is especially unsatisfying in light of the
recent trend to implement path checking algorithms in hardware, which is
inherently parallel.  For example,  the IEEE standard temporal logic
PSL~\cite{IEEE/07/Standard}, an extension of \LTL, has become part of the
hardware description language VHDL, and several
tools~\cite{Dahan+others/05/Combining,BouleZilic2008,FinkbeinerKuhtz09} are available to
synthesize hardware-based monitors from assertions written in PSL.  Can we
improve over the sequential approach by evaluating entire blocks of path
positions in parallel?

\paragraph{\bf Parallelizing \LTL\ path checking}
We show that \LTL\ path checking can indeed be parallelized
efficiently.  Our approach is inspired by work in the related area of evaluating
monotone Boolean
circuits~\cite{Goldschlager80,DymondCook89,Kosaraju90,BarringtonLuMiltersenSkyum99,LimayeMahajanJayalalSarma06,ChakrabortyDatta06}.
Rather than sequentially traversing the path, we consider the circuit that
results from unrolling the formula in positive normal form over the path using
the expansion laws of the logic. Using the positive normal form of the formula
ensures that the resulting circuit is monotone. 
The size of the circuit is quadratic in the size of formula plus the size of the
path. For logarithmic measures, the circuit thus is of the same order as the
input.
Figure~\ref{fig:simple} shows such a circuit for the formula $\left(\left(
a\until b \right) \until \left( c \until d \right)\right) \until e$ and a path
of length 5.

\input{figure-simple}

Yang \cite{Yang91} and, independently, Delcher and Kosaraju
\cite{DelcherKosaraju95} have shown that monotone Boolean circuits can be
evaluated efficiently in parallel \emph{if the graph of the circuit has a planar
embedding}. Unfortunately, this condition is already violated in the simple
example of Figure~\ref{fig:simple} as shown in Figure~\ref{fig:simple-non-planar}.
Individually, however, each operator results in a planar circuit: for example,
$d\,\until\,e$ results in $e_0 \vee (d_0\wedge (e_1 \vee (d_1 \wedge
\ldots)\cdots)$. The complete formula thus defines  a tree of planar circuits. 

\input{figure-simple-non-planar}

Our path checking algorithm works on this tree of circuits. 
We perform a parallel tree contraction
\cite{AbrahamsonEtAl89,KosarajuDelcher88,KarpRamachandran90} to collapse a
parent node and its children nodes into a single planar circuit.  
Simple paths in the tree immediately collapse into a planar circuit; the
remaining binary tree is contracted incrementally, until only a single planar
circuit remains.
The key insight for this construction is that a contraction can be carried out
\emph{as soon as one of the children has been evaluated}.
Initially, all leaves correspond to atomic propositions. 
During the contraction all leaves are evaluated.
Because no leaf has to wait for the evaluation of its sibling before it can be
contracted with its parent, we can contract a fixed portion of the nodes in
every sequential step, and therefore terminate in at most a logarithmic number
of steps.

The path checking problem can, hence, be parallelized efficiently.
The key properties of \LTL\ that are exploited in the construction are the
existence of a positive normal form of linear size and expansion laws that, when
iteratively applied, increase the size of the Boolean circuit only
linearly in the number of iteration steps.
The combinatorial structure of the resulting circuit allows for an
efficient reduction of the evaluation problem to the evaluation problem of
planar Boolean circuits.
In addition to planarity, our construction maintains some further technical
invariants, in particular that the circuits have all input gates on the outer
face.
Analyzing this construction, we obtain the result that the path checking problem
is in~$\ACO(\logDCFL)$:

\begin{thm2}\label{thm:mainltl} 
  The \LTL\ path checking problem is in $\ACO(\logDCFL)$.  
\end{thm2}

The $\ACO(\logDCFL)$ complexity results from a reduction that is performed by
an outer algorithm, which can be implemented as a uniform family of Boolean
circuits of logarithmic depth, and an inner operation, which is represented through
unbounded fan-in gates that are embedded in the circuits and that serve as
$\logDCFL$ oracles.
The $\ACO$ complexity of the outer reduction is due to the tree contraction
algorithm.  Throughout the contraction, each atomic contraction operation
processes a sub-circuit with a size that is of the same order as the overall input
circuit.  Hence, the oracle gates have non-constant fan-in. The whole tree
contraction is performed in a logarithmic number of parallel steps. Thus, the
contraction algorithm implements an $\ACO$ reduction.
Within the $\ACO$ contraction circuits, the oracle gates perform the evaluation
of a certain class of monotone planar Boolean circuits. This operation can be
bound to a complexity of $\logDCFL$.
In summary, the overall path checking algorithm consists of two sequential
reduction steps: First the \LTL\ path checking problem is reduced (in logarithmic
space) to the problem of evaluating a certain class of monotone Boolean
circuits. Second, the problem of evaluating those circuits is $\ACO$ reduced to
the problem of evaluating a certain class of monotone planar Boolean circuits.
The latter problem is solved by an $\logDCFL$ oracle.

The \LTL\ path checking problem is closely related to the membership
problems for the various types of regular expressions: the membership problem
is in \NL\ for regular expressions~\cite{JiangRavikumar91}, in \logCFL\
for semi-extended regular expressions~\cite{Petersen02}, and
\poly-complete for star-free regular expressions and extended regular
expressions~\cite{Petersen00}.  Of particular interest is the
comparison to the star-free regular expressions, since they have the
same expressive power as \LTL\ on finite
paths~\cite{LichtensteinPnueliZuck85}. With $\ACO(\logDCFL)$
vs. $\poly$, our result demonstrates a computational advantage for
\LTL.

\paragraph{\bf \LTL\ with past and bounded-future operators}

Practical temporal logics like PSL 
extend \LTL\  with additional operators that help the user to write shorter
and simpler specifications.  Such extensions often come at a price: adding
extended regular expressions, for example, makes the path checking problem
\poly-complete~\cite{Petersen00}. We show that this is not always the case:
past-time and bounded operators are two major extensions of \LTL, which both
improve the succinctness of the logic exponentially, and whose path checking
problems remain efficiently parallelizable.

\emph{Past-time operators} are the dual of the standard modalities, referring to
past instead of future events. Past-time operators greatly simplify properties
like ``$b$ is always preceded by $a$'', which, in the core logic, require an
unintuitive application of the Until operator, as in $\G \neg (\neg a \until b
\wedge \neg a)$.  Furthermore, Laroussinie, Markey and
Schoebelen~\cite{LaroussinieMarkeySchnoebelen02} proved that the property ``all
future states that agree with the initial state on propositions $p_1, p_2,
\ldots p_n$, also agree on proposition $p_0$,'' which can obviously be expressed
as a simple past-time formula, requires an exponentially larger formula if only
future-time operators are allowed. However, since past operators are the dual of
future operators, they also result in planar circuits; hence, the construction
for \LTL\ can directly be applied to the tree of circuits that results
from \LTL\ formulas with unbounded past and future operators and we obtain the
following result: 

\begin{thm2}\label{thm:mainpltl} 
  The \PLTL\ path checking problem is in $\ACO(\logDCFL)$.  
\end{thm2}

\emph{Bounded operators} express that a condition holds \emph{at least} for a
given, fixed number of steps, or must occur \emph{within} such a number of
steps. Bounded specifications are especially useful in monitoring applications
\cite{FinkbeinerKuhtz09}, where unbounded modalities are problematic: if only
the finite prefix of a computation is visible, it is impossible to falsify an
unbounded liveness property or validate an unbounded safety property. The
succinctness of the bounded operators is due to the fact that expanding the
bounded operators into a formula tree replicates subformulas, causing an
exponential blow-up in the formula size.
Another exponential blow-up is due to the logarithmic encoding of the bounds
compared to an unary encoding in the form of nested next-operators.

A naive solution for the path checking problem of
the extended logic would be to simply unfold the formula to the core
fragment and then apply the construction described above for the
\LTL\ operators. Because of the doubly exponential blow-up, however, such a
solution would no longer be in $\NC$.
If we instead apply the expansion laws for the bounded operators to
the original formula, we obtain a circuit of polynomial size, but with a
more complex structure. Because of this more complex structure, the
path checking construction that we described above for the core logic
is no longer applicable.
Consider the circuit corresponding to the bounded formula $\phi \buntil{3}
\psi$, shown in Figure~\ref{fig:non-planar}: Since the graph of the circuit
contains a $K_{3,3}$ subgraph, it has no planar embedding. Translating a formula
with bounded operators to a tree of circuits would thus include non-planar
circuits, which in general cannot be evaluated efficiently in parallel.

\input{figure-non-planar}

The key insight of the construction for the extended logic is that, although the
circuit for the bounded operators is not planar \emph{a priori}, an
\emph{equivalent} planar circuit can be constructed \emph{as soon as one of the
direct subformulas has been evaluated.}  
Suppose, for example, that the $\phi_i$-gates in the circuit shown in
Figure~\ref{fig:non-planar} are constants. Propagating these constants
\emph{eliminates} all edges that prevent the shown embedding from being planar!
In general, simple propagation is not enough to make the circuit planar. This is
illustrated in Figure~\ref{fig:planar}, where the same formula is analyzed under
the assumption that the $\psi_i$-gates are constant.  While the propagation of
the constants replaces parts of the circuit (identified by the dotted lines)
with constants, there remain references to $\phi_i$-gates, e.g., the two
references to $\phi_2$, that prevent the shown embedding from being planar.
However, an equivalent planar circuit exists: This circuit, shown in
Figure~\ref{fig:planar} as a gray overlay, replaces the disturbing references to
the $\phi_i$-gates by vertical edges to subcircuits. For example, the first
occurrence of $\phi_2$ in $\phi_2 \wedge (0 \vee (\phi_2 \wedge (0 \vee (\psi_3
\wedge 1))$ is replaced with an edge to the subcircuit $\phi_2 \wedge (0 \vee
(\psi_3 \wedge 1))$. The resulting circuit is equivalent, because the additional
conjunct is redundant. 

\input{figure-planar}

Based on these observations, we present a translation from bounded temporal
formulas to circuits that is guaranteed to produce planar circuits, but requires
that one of the direct subformulas has already been evaluated. To meet this
requirement, our path checking algorithm generates the circuits
\emph{on-the-fly}: a circuit for a subformula $\phi$ is constructed only when
a direct subformula of $\phi$ is already evaluated. In this way, we avoid the
construction of circuits that cannot be evaluated efficiently in parallel. As
in the algorithm for \LTL, we evaluate a fixed portion of the subformulas in
every sequential step and thus terminate in time logarithmic in the size of the
formula (bounds are encoded in $O(1)$) plus the length of the path. We prove the
following result:

\begin{thm2}\label{thm:main} 
  The \BLTL\ path checking problem is in $\ACO(\logDCFL)$.  
\end{thm2}

The remainder of the paper is structured as follows. After
preliminaries in Section~\ref{sec:preliminaries},
Section~\ref{sec:monotonecircuits} discusses the evaluation of
monotone Boolean circuits. We introduce \emph{transducer circuits},
which are circuits with a defined interface of input and output gates,
and show that the composition of two planar transducer circuits can be
computed in $\logDCFL$. In Section~\ref{sec:ltlcircuit}, we describe
the on-the-fly translation of \BLTL-formulas to planar circuits.  In
Section~\ref{sec:contraction}, we present the parallel path checking
algorithm. We conclude with pointers to open questions in
Section~\ref{sec:conclusions}.



\section{Preliminaries}
\label{sec:preliminaries}

\subsection{Linear-Time Temporal Logic}
\noindent We consider linear-time temporal logic (\LTL) with 
the usual finite-path semantics, which includes a weak and a strong
version of the Next operator~\cite{LichtensteinPnueliZuck85}.
Let $P$ be a set of atomic propositions. The \emph{LTL formulas}  are defined inductively as
follows: every atomic proposition $p \in P$ is a formula. If $\phi$ and $\psi$
are formulas, then so are 
\[ \neg \phi, \quad \phi\wedge\psi, \quad \phi\vee\psi, \quad \Next\, \phi,  \quad \WNext\,
\phi, \quad \phi\,\until\,\psi, \quad \mbox{ and}\quad  \phi\release\psi\enspace.\]
The size of a formula $\phi$ is denoted by $\abs{\phi}$.

\LTL\ formulas are evaluated over computation paths. A \emph{path}
$\rho=\rho_{0},\dots,\rho_{n-1}$ is a finite sequence of states where each
\emph{state} $\rho_{i}$ for $i=0,\dots,n-1$ is a valuation $\rho_{i} \in 2^P$ of
the atomic propositions. The \emph{length} of $\rho$ is $n$ and is denoted by
$\abs{\rho}$.  Given an LTL formula $\phi$, a nonempty path $\rho$ satisfies
$\phi$ at position $i$ ($0 \leq i < \abs{\rho}$), denoted by
$(\rho,i)\models\phi$, if one of the following holds:

\begin{iteMize}{$\bullet$}
  \item $\phi\in P$ and $\phi \in \rho_{i}$, 
  \item $\phi = \neg \psi$ and $(\rho,i)\not\models\psi$, 
  \item $\phi=\phi_{l}\wedge\phi_{r}$ and $(\rho,i)\models\phi_{l}$ and $(\rho,i)\models\phi_{r}$,
  \item $\phi=\phi_{l}\vee\phi_{r}$ and $(\rho,i)\models\phi_{l}$ or $(\rho,i)\models\phi_{r}$,
  \item $\phi=\Next\psi$ and $i+1 < \abs{\rho}$ and $(\rho,i+1)\models\psi$,
  \item $\phi=\WNext\psi$ and $i+1=\abs{\rho}$ or $(\rho,i+1)\models\psi$,
  \item $\phi=\phi_{l} \until \phi_{r}$ and $\exists j, i \leq j < \abs{\rho} \qdot (\rho,j) \models \phi_{r}$ and $\forall k, i \leq k < j\qdot (\rho,k) \models \phi_{l}$, or
  \item $\phi=\phi_{l} \release \phi_{r}$ and $\forall j, i \leq j < \abs{\rho}\qdot (\rho,j) \models \phi_r$ or $\exists k, i \leq k < j \qdot (\rho,k) \models \phi_{l}$.
\end{iteMize}

\noindent
An LTL formula $\phi$ is \emph{satisfied} by a nonempty path $\rho$ (denoted by $\rho \models \phi)$ iff $(\rho,0) \models \phi$.
By $\phi(\rho)$ we denote the Boolean
sequence $s \in \Bool^{\abs{\rho}}$ with $s_i = 1$ if and only if $(\rho,i) \models
\phi$ for $0 \leq i < \abs{\rho}$.

An \LTL\ formula $\phi$ is said to be in \emph{positive normal form} if in
$\phi$ only atomic propositions appear in the scope of the symbol $\neg$. 
The following dualities ensure that each \LTL\ formula $\phi$ can be rewritten
into a formula $\phi'$ in positive normal form with $\abs{\phi'} =
O(\abs{\phi})$.

\begin{align*}
  \neg \neg \phi\, &\equiv\, \phi\enspace;\\
  \neg \WNext \phi\, &\equiv\, \Next \neg \phi\enspace;\\
  \neg (\phi_{l} \wedge \phi_{r}) \, &\equiv\, (\neg \phi_{l}) \vee (\neg \phi_{r})\enspace;\\
  \neg (\phi_{l} \until \phi_{r}) \, &\equiv\, (\neg \phi_{l}) \release (\neg \phi_{r})\enspace.
\end{align*}

\noindent 
The semantics of \LTL\ implies the \emph{expansion laws}, which relate the
satisfaction of a temporal formula in some position of the path to the
satisfaction of the formula in the next position and the satisfaction of its
subformulas in the present position:

\begin{align*}
\phi_{l} \until \phi_{r}\, &\equiv\, \phi_r \vee (\phi_l \wedge \Next\,(\phi_{l} \until \phi_{r}))\enspace;\\
\phi_{l} \release \phi_{r}\, &\equiv\, \phi_r \wedge (\phi_l \vee \WNext\,(\phi_{l} \release \phi_{r}))\enspace.
\end{align*}

\noindent 
We now extend \LTL\ with the \emph{past-time operators} $\pNext$ (strong
Yesterday), $\pWNext$ (weak Yesterday), $\puntil$ (Since), and $\prelease$
(Trigger) with the following semantics:

\begin{iteMize}{$\bullet$}
  \item$(\rho,i) \models \pNext\psi$ iff 
    \begin{equation*}
      i-1 \geq 0 \text{ and } (\rho,i-1)\models\psi\enspace,
    \end{equation*}
  \item$(\rho,i) \models \pWNext\psi$ iff 
    \begin{equation*}
      i-1 < 0 \text{ or } (\rho,i-1)\models\psi\enspace,
    \end{equation*}
  \item$(\rho,i) \models \phi_{l} \puntil \phi_{r}$ iff 
    \begin{equation*}
      \exists j, i \geq j \geq 0 \qdot (\rho,j) \models \phi_{r} \text{ and } \forall k, i \geq k > j \qdot (\rho,k) \models \phi_{l}\enspace, \text{ and}
    \end{equation*}
  \item$(\rho,i) \models \phi_{l} \prelease \phi_{r}$ iff 
    \begin{equation*}
      \forall j, i \geq j \geq 0 \qdot (\rho,j) \models \phi_r \text{ or } \exists k, i \geq k > j \qdot (\rho,k) \models \phi_{l}\enspace.
    \end{equation*}
\end{iteMize}
\noindent
We call the resulting logic \emph{linear-time temporal logic with past (\PLTL)}.
The following dualities ensure that each \PLTL\ formula $\phi$ can be rewritten
into a formula $\phi'$ in positive normal form with $\abs{\phi'} =
O(\abs{\phi})$.

\begin{align*}
  \neg \pWNext \phi\, &\equiv\, \pNext \neg \phi\enspace;\\
  \neg (\phi_{l} \puntil \phi_{r}) \, &\equiv\, (\neg \phi_{l}) \prelease (\neg \phi_{r})\enspace.
\end{align*}

\noindent The expansion laws for the past operators are
\begin{align*}
  \phi_{l} \puntil \phi_{r}\, &\equiv\, \phi_r \vee (\phi_l \wedge \pNext\,(\phi_{l} \puntil \phi_{r}))\enspace;\\
  \phi_{l} \prelease \phi_{r}\, &\equiv\, \phi_r \wedge (\phi_l \vee \pWNext\,(\phi_{l} \prelease \phi_{r}))\enspace.
\end{align*}

\noindent To obtain \emph{linear-time temporal logic with past and bounds
(\BLTL)} we further add the bounded temporal operators $\buntil{b}$,
$\brelease{b}$, $\pbuntil{b}$, and $\pbreleaseNoblank{b}$, where $b \in \Nat$ is any
natural number. (For technical reasons, the size of a formula is defined using
unary encoding for the bounds. However, our results are actually indepent of
the encoding of the bounds.) The semantics of the bounded operators is defined as
follows:

\begin{iteMize}{$\bullet$}
  \item $(\rho,i) \models \phi_{l} \buntil{b} \phi_{r}$ iff 
    \begin{equation*}
      \exists j, i \leq j \leq \min(i+b,\abs{\rho}-1) \qdot (\rho,j) \models \phi_{r} \text{ and } \forall k, i \leq k < j \qdot (\rho,k) \models \phi_{l}\enspace, 
    \end{equation*}
  \item $(\rho,i) \models \phi_{l} \brelease{b} \phi_{r}$ iff 
    \begin{equation*}
      \forall j, i \leq j \leq \min(i+b,\abs{\rho}-1) \qdot (\rho,j) \models \phi_r \text{ or } \exists k, i \leq k < j \qdot (\rho,k) \models \phi_{l}\enspace,
    \end{equation*}
  \item $(\rho,i) \models \phi_{l} \pbuntil{b} \phi_{r}$ iff 
    \begin{equation*}
      \exists j, i \geq j \geq \max(i-b,0) \qdot (\rho,j) \models \phi_{r} \text{ and } \forall k, i \geq k > j \qdot (\rho,k) \models \phi_{l}\enspace, \text{ and}
    \end{equation*}
  \item $(\rho,i) \models \phi_{l} \pbrelease{b} \phi_{r}$ iff 
    \begin{equation*}
      \forall j, i \geq j \geq \max(i-b,0) \qdot (\rho,j) \models \phi_r \text{ or } \exists k, i \geq k > j \qdot (\rho,k) \models \phi_{l}\enspace. 
    \end{equation*}
\end{iteMize}

\noindent The following dualities apply to the \BLTL\ operators:
\begin{align*}
  \neg (\phi_{l} \buntil{b} \phi_r) \, &\equiv\, (\neg \phi_l) \brelease (\neg \phi_r)\enspace;\\
  \neg (\phi_{l} \pbuntil{b} \phi_r) \, &\equiv\, (\neg \phi_l) \pbrelease (\neg \phi_r)\enspace.
\end{align*}
\noindent The expansion laws for the bounded operators are defined as follows for $b \in \Nat$:
\begin{align*}
&\phi_{l} \buntil{b} \phi_{r}\,\equiv 
  \begin{cases}
    \phi_r \vee (\phi_l \wedge \Next\,(\phi_{l} \buntil{b-1} \phi_{r})) &\text{for } b > 0,\\
    \phi_r &\text{for } b = 0,
  \end{cases}\\
&\phi_{l} \brelease{b} \phi_{r}\,\equiv 
  \begin{cases}
    \phi_r \wedge (\phi_l \vee \WNext\,(\phi_{l} \brelease{b-1} \phi_{r})) &\text{for } b > 0,\\
    \phi_r &\text{for } b = 0,
  \end{cases}\\
&\phi_{l} \pbuntil{b} \phi_{r}\,\equiv 
  \begin{cases}
    \phi_r \vee (\phi_l \wedge \pNext\,(\phi_{l} \pbuntil{b-1} \phi_{r})) &\text{for } b > 0,\\
    \phi_r &\text{for } b = 0, \text{ and }
  \end{cases}\\
&\phi_{l} \pbrelease{b} \phi_{r}\,\equiv 
  \begin{cases}
    \phi_r \wedge (\phi_l \vee \pWNext\,(\phi_{l} \pbrelease{b-1} \phi_{r})) &\text{for } b > 0,\\
    \phi_r, &\text{for } b = 0.
  \end{cases}
\end{align*}

\noindent
We are interested in determining if a formula is satisfied by a given
path. This is the path checking problem.

\begin{defi}[Path Checking Problem]
The path checking problem for \LTL\ (\PLTL, \BLTL) is to decide, for an \LTL\
(\PLTL, \BLTL) formula $\phi$ and a nonempty path $\rho$, whether $\rho\models\phi$.
\end{defi}

Later in this paper we will present a path checking algorithm for
\BLTL. 
The algorithm constructs a circuit that is of polynomial size in the length of
the input computation path and in the size of the input
formula including the sum of the bounds.
However, we do not want the complexity of the algorithm to depend on the
encoding of the bounds.
The following lemma allows us to prune the size of the bounds that occur in a
\BLTL\ formula to the length of the computation path.

\begin{lem}\label{lem:bltlprune}\sloppy
  Given a \BLTL\ formula $\phi$ and a finite computation path $\rho$, the \BLTL\ formula
  $\phi'$ is obtained from $\phi$ by setting each bound $n$ in $\phi$ to
  $\min(n,\abs{\rho})$. It holds that $\rho \models \phi$ if and only if $\rho
  \models \phi'$.
\end{lem}
\begin{proof}
  By induction over $\phi$.
\end{proof}

\subsection{Complexity classes within \texorpdfstring{\poly}{\textsf{P}}.}
We assume familiarity with the standard complexity classes within \poly.
$\NC$ is the class of decision problems decidable in polylogarithmic time on a parallel computer with a polynomial number of processors.
\logspace~is the class of problems that can be decided by a logspace
restricted deterministic Turing machine. 
\logDCFL\ is the class of problems that can be decided by a logspace and
polynomial time restricted deterministic Turing machine that is additionally equipped with a
stack. 
$\AC^i, i \in \Nat$, denotes the class of problems decidable by polynomial size
unbounded fan-in Boolean circuits of polylogarithmic depth of degree $i$. \AC\ is defined as
$\bigcup_{i\in\Nat}\AC^i$.  Throughout the paper, all circuits are assumed to be
uniform. Often we use functional versions of complexity classes.
Since in our case the output size of the functions is always polynomially
bounded we can use a polynomial number of circuits for the corresponding class
of decision problems, each for computing a single bit of the output. 
Thus, in the following we do not explicitly distinguish between decision
problems and functional problems \cite{Johnson90}.  
It holds that 
\begin{equation*}
\AC^0 \subsetneq \logspace \subseteq \logDCFL \subseteq \AC^1 \subseteq \AC^2
  \subseteq \dots \subseteq \AC=\NC \subseteq \poly\enspace.
\end{equation*}

\noindent Further details can be found in the survey paper by Johnson \cite{Johnson90}.

Given a problem $P$ and a complexity class $C$, $P$ is $\AC^1$ Turing reducible
to $C$ (denoted as $P \in \AC^1(C)$) if there is a family of $\AC^1$ circuits
with additional unbounded fan-in $C$-oracle gates that decides $P$. 
It holds that
\begin{equation*}
\AC^1 \subseteq \AC^1(\logDCFL) \subseteq \AC^2\enspace.
\end{equation*}

\noindent For further details on $\ACO$ reductions, we refer to \cite{Vollmer99}.

\subsection{Parallel Tree Contraction}

The path checking algorithm presented in this paper relies on efficient
parallel tree contraction. Here we follow the approach of \cite{AbrahamsonEtAl89}
and \cite{KosarajuDelcher88}. 
%
%
%
%
%
A rooted binary tree is called \emph{regular} if all inner nodes have exactly two
children.
Let $\T_0=\langle V_0, E_0 \rangle$ be an ordered, rooted, regular, binary tree. 
A contraction step on $\T_i$ takes a leaf $l$ of $\T_i$, its sibling $s$,
and its parent $p$ and contracts these nodes into a single node $s'$ in the tree
$\T_{i+1} = \langle V_{i+1}, E_{i+1} \rangle$ with
\begin{align*}
  V_{i+1} &= \left(V_i \setminus \{l,p\}\right), \text{ and }\\
  E_{i+1} &=  \begin{cases}
    E_i \setminus \{\langle p,l \rangle,\langle p,s\rangle\} \quad&\text{if $p$ is the root of $\T_i$},\\
    \left(E_i \setminus \{\langle p,l \rangle,\langle p,s \rangle,
  \langle q,p \rangle \}\right)\cup \{ \langle q,s\rangle \}, \text{$q \in V_i$ is parent of $p$} \quad &\text{otherwise}.
  \end{cases}
\end{align*}
Using the fact that a contraction step is a local operation it is
possible to perform contraction steps in parallel on non-overlapping subtrees.

A tree contraction on an ordered, rooted, regular, binary tree $T$ is a process that
iteratively applies contraction steps on the tree $T$ until it is contracted
into a singleton tree.
\prettyref{alg:contraction} from \cite{KarpRamachandran90} performs a tree
contraction in $\ceil{\log n}$ stages of parallel contraction steps.

\newcommand\Input{\mbox{\textbf{Input:}}}
\newcommand\Effect{\mbox{\textbf{Effect:}}}

\begin{algo}\hfill\\
\Input~an ordered, rooted, regular, binary tree $T$ with $n$ leaves.\\
\Effect~contracts $T$ into a singleton tree.
\label{alg:contraction}
\begin{algorithmic}
  \STATE Number the leaves in order from left to right as $1,\dots,n$.
  \FOR{$\ceil{\log n}$ iterations}
    \STATE Apply the contraction step to all odd numbered leaves that are the
    left child of their parent.
    \STATE Apply the contraction step to all odd numbered leaves that are the
    right child of their parent.
    \STATE Update the numbering of the remaining (even numbered) leaves by
    dividing each leaf number by two.
  \ENDFOR
\end{algorithmic}
\end{algo}

\noindent The algorithm can be implemented on an \emph{exclusive read exclusive
write random access memory machine (EREW PRAM)} such that it runs in
time $O(\log n)$ with a total work of $O(n)$ \cite{KarpRamachandran90}. 
It is well known that problems that can be solved on an EREW PRAM in time
$O(\log n)$ with polynomial total work are contained in $\ACO$ \cite{Vollmer99}. 
\prettyref{fig:contraction} shows a tree contraction process for an
example tree.

\input{figure-contraction}

In order to use the parallel tree contraction algorithm to compute some function
$f$ on a labeled tree, the contraction step is piggybacked with a local
operation on the labels of the nodes involved in the contraction step. 
In order for $f$ to be in $\ACO$, the individual contraction steps must be performed
in constant time.
%
%
For our constructions this is not the case. 
However, by piggybacking the contraction step with operations that for some
complexity class $\mathcal{C}$ are solvable with $\mathcal{C}$-oracle gates, the
problem of computing $f$ is $\ACO$-reduced to $\mathcal{C}$. 
Hence, by showing that the complexity of the contraction step is in $\mathcal{C}$, the
overall complexity of $f$ is proven to be in $\ACO(\mathcal{C})$.

\section{Monotone Boolean Circuits}
\label{sec:monotonecircuits}

\noindent
A \emph{monotone Boolean circuit} $\langle \Gamma, \gamma \rangle$ consists of a
set $\Gamma$ of \emph{gates} and a gate labeling $\gamma$.  The \emph{gate
labeling} labels each gate either with a Boolean value, with the symbol $\gvar$, with a
tuple $\langle \gop,\gleft,\gright\rangle$, or with a tuple $\langle \gaid,\gsuc
\rangle$, where $\gop\in\{\gand,\gor\}$, and $\gleft$, $\gright$, and $\gsuc$
are gates.

A gate that is labeled with a Boolean value is called a \emph{constant gate}.  A
gate that is labeled with $?$ is called a \emph{variable gate}.  For a
non-constant, non-variable gate $a$ labeled with $\langle \gop,b,c\rangle$ or
$\langle \gaid,b \rangle$, we say that $a$ \emph{directly depends} on $b$ and
$c$, denoted by $a \ddep b$, $a \ddep c$. The \emph{dependence} relation is the
transitive closure of $\ddep$.  A gate on which no other gate depends  is called
a \emph{sink gate}. \emph{A circuit must not contain any cyclic dependencies.}
 
For a circuit  $G = \langle \Gamma, \gamma \rangle$, $\constants(G)$
denotes the set of all constant gates in $\Gamma$. If $\Gamma =
\constants(G)$, we call $G$ \emph{constant}. By $\vars(G)$ the set of all
variable gates of $\Gamma$ is denoted. Finally we define $\inputs(G)$ to be the set
of all variable gates and all constant gates that are not sink gates in
$\Gamma$.
In the following, we assume that all circuits are monotone Boolean circuits. We
omit the labeling whenever it is clear from the context and identify the
circuit with its set of gates.  

\subsection{Circuit evaluation}
The \emph{evaluation} of a circuit $\langle \Gamma, \gamma \rangle$ is the
  (unique) circuit $\langle \Gamma, \gamma' \rangle$ where for each gate $g \in
  \Gamma$ the following holds:
  \begin{iteMize}{$\bullet$}
    \item $\gamma'(g) = 0$ iff $\gamma(g) = \langle and,l,r\rangle$ and
      $\gamma'(l)=0$ or $\gamma'(r)=0$,
    \item $\gamma'(g) = 1$ iff $\gamma(g) = \langle and,l,r\rangle$ and
      $\gamma'(l)=1$ and $\gamma'(r)=1$,
    \item $\gamma'(g) = \langle id,l \rangle$ iff $\gamma(g) = \langle
      and,l,r\rangle$ and $\gamma'(l)\not\in \{0,1\}$ and $\gamma'(r) = 1$,
    \item $\gamma'(g) = \langle id,r \rangle$ iff $\gamma(g) = \langle
      and,l,r\rangle$ and $\gamma'(r)\not\in \{0,1\}$ and $\gamma'(l) = 1$,
    \item $\gamma'(g) = 0$ iff $\gamma(g) = \langle or,l,r\rangle$ and
      $\gamma'(l)=0$ and $\gamma'(r)=0$,
    \item $\gamma'(g) = 1$ iff $\gamma(g) = \langle or,l,r\rangle$ and
      $\gamma'(l)=1$ or $\gamma'(r)=1$,
    \item $\gamma'(g) = \langle id,l\rangle$ iff $\gamma(g) = \langle
      or,l,r\rangle$ and $\gamma'(l)\not\in \{0,1\}$ and $\gamma'(r) = 0$,
    \item $\gamma'(g) = \langle id,r\rangle$ iff $\gamma(g) = \langle
      or,l,r\rangle$ and $\gamma'(r)\not\in \{0,1\}$ and $\gamma'(l) = 0$,
    \item $\gamma'(g) = \gamma'(s)$ iff $\gamma(g) = \langle id,s\rangle$ and
      $\gamma'(s) \in \{0,1\}$, and
    \item $\gamma'(g)=\gamma(g)$ otherwise.
  \end{iteMize}
  A circuit is \emph{evaluated} if all constant gates are sink gates.
In an evaluated circuit, all gates that do not depend on variable gates are
  constant.  Hence, a circuit without any variable gates evaluates to a constant
  circuit; for a circuit that contains variable gates, a subset of the gates is
  relabeled: some $\gand$-/$\gor$-/$\gaid$-gates are labeled as constant or
  $\gaid$-gates.  

  The problem of evaluating monotone planar circuits has been studied
  extensively in the literature
  \cite{Goldschlager80,DymondCook89,Kosaraju90,BarringtonLuMiltersenSkyum99,LimayeMahajanJayalalSarma06,ChakrabortyDatta06}.
  Our construction is based on the evaluation of one-input-face planar circuits:
Given a circuit $G=\langle \Gamma, \gamma \rangle$ with variable gates $X$, the
\emph{graph} $\gr(G)$ \emph{of} $G$ is the directed graph $\langle \Gamma
,E\rangle$, where $E = \{\langle a,b\rangle \in \Gamma \times \Gamma \mid a
\ddep b\}$.  A circuit~$C$ is \emph{planar} if there exists a planar embedding
of the graph of $C$. A planar circuit $G$ is \emph{one-input-face} if there is a
planar embedding such that all gates of $\inputs(G)$ are located on the outer face.
%
Note that
an evaluated planar circuit with all variable gates on the outer face is one-input-face. 
The evaluation of one-input-face planar circuits can be parallelized efficiently.  We make
use of a result by Chakraborty and Datta~\cite{ChakrabortyDatta06}:

\begin{thm}[Chakraborty and Datta 2006]\label{thm:ChakrabortyDatta2006} 
  The problem of evaluating an one-input-face planar circuit without variable gates is in \logDCFL.
\end{thm}

\noindent
Using standard techniques \cite{Kosaraju90}, the theorem
generalizes to circuits that contain variable gates:

\begin{cor}\label{cor:oifpeceval}
The problem of evaluating an one-input-face planar circuit is in \logDCFL.  
\end{cor}

\begin{proof}
We first assign the Boolean constant \one\ to all variable gates. Each gate that
evaluates to \zero\ is turned into a \zero\ constant gate. Next, we assign
\zero\ to all variable gates.  Each gate that evaluates to \one\ is turned into a
constant gate with value \one. Since the values of the remaining gates depend on the
variables, they are simply copied. If one of the latter gates depends on a
constant gate, the dependency is removed by changing such a gate into an
$\gaid$-gate.
\end{proof}
 
\subsection{Transducer Circuits}

\noindent The central construction in our path checking algorithm is circuit composition: 
circuits for larger subformulas are built from circuits for smaller 
subformulas by connecting variable gates of one circuit to gates of another
circuit. To facilitate this operation, we introduce transducer circuits, 
which are circuits with a defined interface of input and output gates that allow
the circuit to transform a sequence of Boolean input values, for example the values of a 
subformula at different positions of the path, into a sequence of output values.

A \emph{transducer circuit} is a tuple $T=\langle \Gamma, \gamma, I, O
\rangle$ where $G=\langle \Gamma, \gamma \rangle$ is a circuit, $I$ is a
(strict) ordering of $\vars(G)$, and $O$ is a (strict) ordering of a subset of $\Gamma$. $I$
is called the input of $T$ and $O$ is called the output of $T$. The \emph{input} and
\emph{output arity} is the length of the input and output, denoted as $\abs{I}$ and $\abs{O}$,
respectively. We denote the $i$\textsuperscript{th} element of $I$ and $O$ by $I(i)$ and $O(i)$, respectively.
The
transducer circuit $T$ is \emph{planar} if $G$ has a planar embedding such that the
gates of $I$ appear counter-clockwise ordered on the outer face, the gates of
$O$ appear clockwise ordered on the outer face, and between any two gates of $I$
on the outer face there are either no or all gates of $O$, i.e., the gates of
$I$ and $O$ do not appear interleaved on the outer face.

Given two planar transducer circuits $G = \langle
\Gamma,\gamma,I_G,O_G\rangle$ and $D = \langle \Delta,\delta,I_D,O_D\rangle$, 
$G$ is \emph{composable with $D$} if the input arity of $D$ equals the output arity of
$G$.
The \emph{composition} $G \ccomp D$ of $G$ with $D$ is the planar transducer circuit $E = \langle
\Epsilon,\epsilon,I_E,O_E \rangle$ with $\Epsilon=\Gamma \udot \Delta$, $I_E =
I_G$, $O_E = O_D$ and 
\begin{align*}
  &\epsilon(g)=
  \begin{cases}
    \gamma(g), &\text{for } g \in \Gamma,\\
    \delta(g), &\text{for } g \in \Delta \setminus \vars(\Delta), \text{ and }
  \end{cases}\\
  &\epsilon(I_D(i)) = \langle \gaid, O_G(i) \rangle, \text{ for } 0 \leq i < \abs{O_G}.
\end{align*}

\noindent
The composition $G\ccomp D$ can be computed by a logspace restricted
deterministic Turing machine. 

A transducer circuit $T$ represents a function $f_T: \Bool^{\abs{I}} \rightarrow \Bool^{\abs{O}}$, 
where $f_T(s)$ for some sequence $s \in  \Bool^{\abs{I}}$  is computed by
evaluating the composition of $T$ with the constant circuit that represents $s$. The values of 
the output gates of the resulting constant circuit define the sequence $f_T(s)$.

\begin{lem}\label{lem:oifccomp}
  For two one-input-face planar transducer circuits $G$ and $D$, such that $G$
  is composable with
  $D$, the evaluation of $G \ccomp D$ is an evaluated one-input-face planar transducer circuit
  and can be computed within $\logDCFL$.
\end{lem}

\begin{proof}
  Let $G=\langle \Gamma, \gamma, I_G, O_G\rangle$ and $D=\langle \Delta, \delta,
  I_D, O_D \rangle$. 
  The transducer circuits $G'$ and $D'$ are obtained from $G$ and $D$,
  respectively, as follows.  For each $i$, $0 \leq i < \abs{O_G}$, with
  $O_G(i)$ being a constant gate $b$ in $G$
  \begin{iteMize}{$\bullet$}
    \item $b$ is removed from $G'$, including $O'_G(i)$,
    \item $b$ is added to $D'$,
    \item $\delta'(I_D(i)) = b$, and
    \item $I_D(i)$ is removed from $I_D'$.
  \end{iteMize}
  In other words: all constant outputs of $G$ are moved out of $G'$ into $D'$.
  Clearly, $G'$ is composable with $D'$ and the evaluation of $G' \ccomp D'$
  equals the evaluation of $G \ccomp D$.
  Further, $G'$ and $D'$ are both one-input-face and planar.
  Because $G'$ is one-input-face, all constants are either on the outer face or are sinks.
  Since all constant gates that are sinks in $G$ but not in $G \ccomp D$ have
  been moved out of $G'$ into $D'$ it holds that all constants in $G'$ are on
  the outer face or are sinks also in $G' \ccomp D'$.
  $G' \ccomp D'$ generally is not one-input-face, because $\inputs(D')$ can contain
  constant gates that are neither sinks nor on the outer face of $G' \ccomp D'$,
  thus preventing the application of \prettyref{thm:ChakrabortyDatta2006}. 
  However, $D'$ is one-input-face and planar and can thus be evaluated in \logDCFL\ using
  \prettyref{thm:ChakrabortyDatta2006}, resulting in a circuit $D''$ where all
  constants are sinks.
  Now, in the composition of $G'$ with $D''$ all constants are either on the
  outer face of $G'$ or are sinks.
  Thus, $G' \ccomp D''$ is an one-input-face planar transducer circuit that can be
  evaluated in \logDCFL.
  The circuits $G'$, and $D'$, as well as the circuit compositions are
  computable in logarithmic space.
\end{proof}

%
%


\section{Constructing Circuits On-The-Fly}\label{sec:ltlcircuit}

\noindent We now describe the translation of \BLTL-formulas in positive normal
form to planar circuits.
As discussed in the introduction, the translation is not done as a
preprocessing step, but rather delayed until \emph{one} of the direct
subformulas has been evaluated.  We guarantee that the resulting
circuit is planar, one-input-face, and evaluated.
The path checking algorithm, which will be presented in the next
section, composes the evaluated one-input-face planar circuits in order to
represent larger partially evaluated subformulas.

Given a path $\rho$ and a \BLTL\ formula $\phi$ in positive normal form with at
most one unevaluated direct subformula,
the following construction provides a function $\circuit{\rho}$ that maps the
top-level operator of $\phi$ and its evaluated subformulas to an evaluated
one-input-face planar transducer circuit that represents a partial
evaluation of $\phi$ on $\rho$.
The output arity of the circuit is $\abs{\rho}$, the input arity is $\abs{\rho}$
for all formulas except for atomic propositions, where the circuit has input
arity 0. 
The circuit can be constructed by a logspace restricted Turing machine.
The full details of the construction are provided 
in the appendix.

\subsection{Atomic propositions.}
\noindent For an atomic proposition $p$, the circuit is a set of constant gates,
one for each path position. The value of a gate is the value of $p$ at
the respective position of $\rho$:  
$\circuit{\rho}(p) = \langle \left\{ o_0,\dots,o_{n-1} \right\},l,\varepsilon,O \rangle$,
where $n= \abs{\rho}$, $O = o_0, \dots, o_{n-1}$, $l(o_i) = 1$ iff
$p \in \rho_i$, and $\varepsilon$ denotes an empty
input sequence.
Clearly, a set of constant gates is an evaluated one-input-face planar transducer circuit.

\subsection{Unary operators.}
\noindent For the unary operators $\Next, \WNext, \pNext,$ and $\pWNext$, the
circuit shifts the value of the input by one position in the
respective direction. The first (respective last) position of the
output is a constant with value 0 for strong operators and value 1 for
weak operators. Again, the circuits are obviously planar, one-input-face, and
evaluated, and of input
and output arity $\abs{\rho}$.
E.g.\ $\circuit{\rho}(\Next) = \langle G,l,I,O \rangle$, where
$n = \abs{\rho}$, $I = v_0, \dots, v_{n-1}$, $O = o_0, \dots, o_{n-1}$, $G = \{v_0, \dots, v_{n-1}\} \cup \{o_0, \dots, o_{n-1} \}$, and 
\begin{align*}
  l(o_i) &= \begin{cases}
    \langle \gaid, v_{i+1}\rangle &\text{for } 0 \leq i < n-1,\\ 
    0 &\text{for } i = n-1, \text{ and}
  \end{cases}\\
  l(v_i) &= \gvar \text{ for } 0 \leq i < n.
\end{align*}

\subsection{Binary operators.}
\noindent The binary operators require two constructions, one for the case where
the left argument has been evaluated and one for the case where the right
argument has been evaluated.  For each operator $\mathit{op}$, we define
two logspace-computable functions $\circuit{\rho}(s,\mathit{op})$ and
$\circuit{\rho}(\mathit{op},s)$, which compute the circuit given an
evaluation $s \in \Bool^{\abs{\rho}}$ of the left and right
subformula, respectively. 

For the Boolean operators, the two functions are the same, e.g., 
$\circuit{\rho}(\vee,s)= \circuit{\rho}(s,\vee) = \langle G,l,I,O \rangle$, where 
$n = \abs{\rho}$, $I = v_0, \dots, v_{n-1}$, $O = o_0, \dots, o_{n-1}$, $G = \{v_0, \dots, v_{n-1}\} \cup \{o_0, \dots, o_{n-1} \}$, and 
\begin{align*}
  l(o_i) &= \begin{cases}
    \langle \gaid, v_{i}\rangle &\text{for } s_i = 0 \text{ and } 0 \leq i < n,\\ 
    1 &\text{for } s_i = 1 \text{ and } 0 \leq i < n, \text{ and}
  \end{cases}\\
  l(v_i) &= \gvar \text{ for } 0 \leq i < n.
\end{align*}

\noindent 
For the unbounded temporal operators, the constructions are derived
from the expansion laws of the logic, such as $\phi_{l} \until
\phi_{r}\, \equiv\, \phi_r \vee (\phi_l \wedge \Next\,(\phi_{l} \until
\phi_{r}))$ for the unbounded Until operator.  The expanded formula
is transformed into a transducer circuit by
substituting constants for evaluated subformulas and variable gates for
unevaluated subformulas. E.g.\ $\circuit{\rho}(\until,s) = \langle G,l,I,O \rangle$, where 
$n = \abs{\rho}$, $I = v_0, \dots, v_{n-1}$, $O = o_0, \dots, o_{n-1}$, $G = \{v_0, \dots, v_{n-1}\} \cup \{o_0, \dots, o_{n-1} \}$, and 
\begin{align*}
  l(o_{i}) &= \begin{cases}
    \langle \gand, v_i, o_{i+1}\rangle &\text{for } 0 \leq i < n-1 \text{ and } s_i = 0,\\ 
    1 &\text{for } 0 \leq i < n-1 \text{ and } s_i = 1,\\
    s_{n-1} &\text{for } i = n-1, \text{ and}
  \end{cases}\\
  l(v_i) &= \gvar \text{ for } 0 \leq i < n,
\end{align*}
\noindent and $\circuit{\rho}(s,\until) = \langle G,l,I,O \rangle$, where 
\begin{align*}
  l(o_i) &= \begin{cases}
    \langle \gor, v_i, o_{i+1}\rangle &\text{for } 0 \leq i < n-1 \text{ and } s_i = 1,\\
    \langle \gaid, v_i \rangle &\text{for } 0 \leq i < n-1 \text{ and } s_i = 0,\\
    \langle \gaid, v_{n-1} \rangle &\text{for } i = n-1, \text{ and}
  \end{cases}\\
  l(v_i) &= \gvar \text{ for } 0 \leq i < n.
\end{align*}

%

\input{figure-buntil2}

\noindent 
The most difficult part of the construction is the translation for the bounded
operators, which we now present in detail for the bounded Until
operator~$\buntil{b}$. 
%
\prettyref{fig:buntil2} illustrates the construction of
$\circuit{\rho}(s,\buntil{b})$ for a valuation $s = 0,1,0,1,1,1,0,1$ of the left
subformula. The gates indexed by ${i,j}$ compute the value of the formula at
position $i$ and ``remaining'' bound $b-j$.  If, at some position, the left
subformula evaluates to 0, then the formula simplifies to the right subformula,
independently of the remaining bound. This results in vertical edges in the
circuit. If the left subformula evaluates to 1, then the formula is true if it
is either true for bound $j-1$ in position $i+1$ or for bound $j-1$ in position
$i$.  In the circuit, this is computed as a disjunction of the vertical and the
diagonal neighbor.

We define $\circuit{\rho}(s,\buntil{b}) = \langle G,l,I,O \rangle$, where $G=\{v_{i,j} \mid 0 \leq i < n, 0 \leq j \leq b\}$, $I = v_{0,b}, \dots, v_{n-1,b}$, $O = v_{0,0}, \dots, v_{n-1,0}$, and 
\begin{equation*}
  l(v_{i,j}) = \begin{cases}
    \left\langle \gaid,v_{i,j+1} \right\rangle &\text{for } 0 \leq i < n-1 \text{ and } j < b \text{ and } s_i = 0,\\
    \left\langle \gor,v_{i,j+1},v_{i+1,j+1} \right\rangle &\text{for } 0 \leq i < n-1 \text{ and } j < b \text{ and } s_i = 1,\\
    \left\langle \gaid,v_{i,j+1} \right\rangle &\text{for } i = n-1 \text{ and } j < b, \text{ and}\\
    \gvar &\text{for } j = b.
  \end{cases}
\end{equation*} 
\noindent The construction of $\circuit{\rho}(\buntil{3},s)$ for the valuation
$s = 0,1,0,0,0,0,0,1$ of the right subformula is illustrated in
\prettyref{fig:buntil1}. Here, the gates indexed by $i$ compute the value of the
formula at position $i$. If the right subformula evaluates to 1 in position $i$,
then the value of the formula is 1 in position $i$, and is computed by the
conjunction over the values of the left subformulas in positions $i-b$ to $i-1$.
Further to the left from $i-b$, the value is 0 until another 1 occurs in the
valuation of the right subformula.

\input{figure-buntil1}

The circuit is therefore defined as $\circuit{\rho}(\buntil{b},s) = \langle G,l,I,O \rangle$, where  $I = v_0, \dots, v_{n-1}$, $O = o_0, \dots, o_{n-1}$, $G = \{v_0, \dots, v_{n-1}\} \cup \{o_0, \dots, o_{n-1} \}$, and
\begin{align*}
  l(o_{i}) &= \begin{cases}
    1 &\text{for } 0 \leq i < n \text{ and } s_i = 1,\\
    0 &\text{for } 0 \leq i < n \text{ and } \forall j, i \leq j < \min(i+b,n) \qdot s_j = 0,\\
    \langle \gand, v_i, o_{i+1}\rangle &\text{for } 0 \leq i < n-1 \text{, } s_i = 0 \text{, } \exists j,i < j < \min(i+b,n) \qdot s_j = 1, \text{ and}
  \end{cases}\\
  l(v_i) &= \gvar \text{ for } 0 \leq i < n.
\end{align*}

\noindent 
We conclude the section with a lemma that formally states the existence of the
logspace-computable function $\circuit{\rho}$ with the required properties. 
The complete construction of $\circuit{\rho}$ is provided in
the \hyperref[sec:appendix]{appendix}.

\begin{lem}
  Let $\phi$ and $\psi$ formulas and $p$ an atomic proposition.
  Let $\rho$ a path and $s,t \in \Bool^{\abs{\rho}}$ with $s=\phi(\rho)$ and $t
  = \psi(\rho)$.
  There is an logspace-computable function $\circuit{\rho}$ mapping its arguments to
  evaluated one-input-face planar transducer circuits such that
  $p(\rho) = \circuit{\rho}(p)(),$
  $(op\, \phi)(\rho) = \cirf{\circuit{\rho}(op)}(s)$ for $op \in \{\Next,\WNext,\pNext,\pWNext\}$, and
  $(\phi \,op\, \psi)(\rho) = \cirf{\circuit{\rho}(s,op)}(t)$ and
  $(\phi \,op\, \psi)(\rho) = \cirf{\circuit{\rho}(op,t)}(s)$
  for $op \in \{\vee,\wedge,\until,\buntil{b},\release,\brelease{b},\puntil,\pbuntil{b},\prelease,\pbrelease{b}\}$.
  \qed
\label{lem:circuitsound}
\end{lem}
%


\section{Parallel Tree Contraction for Path Checking}
\label{sec:contraction}

\noindent The parallel path checking algorithm for \BLTL\ formulas is based on a
bottom-up evaluation of the formula converted to positive normal form starting with the atomic
propositions. The central data structure is a binary tree, called the
\emph{contraction tree}, that keeps track of the dependencies between
the different evaluation steps.  Initially, the contraction tree
corresponds to the formula tree where each unary node (due to $\Next$,
$\WNext$, $\pNext$, and $\pWNext$ operators) has been merged into a single
node with its unique child node.
The evaluation of the formula is performed by \emph{contraction steps,} which
contract a node that has already been evaluated with its parent into a new edge
from its sibling to its parent. 
The resulting edge is labeled by a planar circuit that represents the partially
evaluated subformula.

Since no child needs to wait for the evaluation of its sibling before
it can be contracted with its parent, a constant portion of the nodes
can be contracted in parallel, and, within logarithmic time, the tree
is evaluated to a single constant circuit.
We now describe and analyze this process in more detail.

\subsection{Contraction tree}

Given a formula $\phi$ in positive normal form and a path $\rho$, let $\phi_0,\dots,\phi_{m-1}$
be the subformulas of $\phi$ with $\phi_0 = \phi$. A contraction tree is an
edge labeled tree $\eT = \langle T,t,l \rangle $ where $T \subseteq
\{\phi_0,\dots,\phi_{m-1}\} \cup \{root\}$, $t \subseteq \left\{ \langle \phi_i,\phi_j\rangle \mid \phi_j
\text{ is a subformula of } \phi_i \right\} \cup \{\langle root,\phi\rangle\}$, and $l$ is a mapping that labels
each edge of $\eT$ with an evaluated one-input-face planar transducer circuit, such that the
following conditions hold:
\begin{enumerate}[(1)]
  \item $\eT \setminus \{root\}$ is an ordered, rooted, regular, binary tree,
  \item all edge labels of $\eT$ are evaluated one-input-face planar transducer circuits of
    arity $\abs{\rho}$, all leaves are atomic propositions, and
  \item for $\tau=\langle \phi_i,\phi_j\rangle \in t$ it holds that
    $\cirf{l(\tau)}(\phi_j(\rho)) = \psi(\rho)$, where $\psi$ is the direct
    subformula of $\phi_i$ that has $\phi_j$ as a subformula.
    Further, for the unique edge $\tau=\langle root,\phi_j\rangle \in t$ it holds that
    $\cirf{l(\tau)}(\phi_j(\rho)) = \phi(\rho)$.
\end{enumerate}
The special node $root$ and the corresponding edge $\langle root,\phi\rangle$
were added solely for technical reasons. 

The first condition ensures that the overall contraction process performs in a
logarithmic number of parallel steps. The second condition provides the
preconditions for a single contraction step. Namely, the compositionality of the
constructed circuits and the complexity of \logDCFL. The third condition states
the induction hypothesis for the soundness of the whole algorithm: When a
transducer circuit is attached to an edge of the contraction tree, it encodes the
semantics of all partially evaluated subformulas contracted into that edge.

\subsection{Initialization step} 

The initial contraction tree $\eT\setminus \{root\}$ is the formula tree $\phi$
where each unary node (due to $\Next$, $\WNext$, $\pNext$, and $\pWNext$
operators) has been merged with its unique child node.
The corresponding new parent edge is labeled by the transducer circuit that results
from composing the circuits produced by applying $\circuit{\rho}$ to the
corresponding \BLTL\ operators of the eliminated nodes. 

\begin{lem}\label{lem:init}
  Given a formula $\phi$ in positive normal form and a path $\rho$, a contraction tree $\eT$ can be
  constructed from $\phi$ and $\rho$ by an logspace restricted Turing machine.
\end{lem}

\newcommand{\idcir}{\ensuremath id}

\begin{proof}
  Define $\parent(\chi)$ to be the subformula $\psi$ of $\phi$ such that
  $\chi$ is the maximal subformula of $\psi$ in $\phi$.
  Let $\eT = \langle T,t,l \rangle$ with 
  $T = \{\phi_i \mid \phi_i \text{ is not of the form } \Next \psi, \WNext \psi,
  \pNext \psi, \text{ or }
  \pWNext \psi, 0 \leq i < m\} \cup \{root\}$, $t = \{\langle \phi_i,\phi_j\rangle \in T \times T \mid \phi_j \text{ is
  a maximal subformula in $T$ of } \phi_i \} \cup \{\langle root,\phi\rangle\}$, and for
  $\tau \in t$,
\begin{equation*}
    l(\tau) =
    \begin{cases}
      \idcir &\text{for $\tau = \langle root,\phi\rangle$},\\
      c(\tau) &\text{otherwise,}
    \end{cases}
\end{equation*}
where for $0 \leq i < m$,
\begin{equation*}
     c(\langle \phi_i,x\rangle) = 
  \begin{cases}
    \circuit{\rho}(\Next) \ccomp c(\parent(\phi_i)) &\text{for } \parent(\phi_i) = \Next \phi_i,\\
    \circuit{\rho}(\WNext) \ccomp c(\parent(\phi_i)) &\text{for } \parent(\phi_i) = \WNext \phi_i,\\
    \circuit{\rho}(\pNext) \ccomp c(\parent(\phi_i)) &\text{for } \parent(\phi_i) = \pNext \phi_i,\\
    \circuit{\rho}(\pWNext) \ccomp c(\parent(\phi_i)) &\text{for } \parent(\phi_i) = \pWNext \phi_i,\\
    \idcir &\text{otherwise,}
  \end{cases}
\end{equation*}
where $\idcir$ is the identity transducer circuit of arity $\abs{\rho}$.

In $\eT$, all simple paths (due to $\Next$, $\WNext$, $\pNext$, and $\pWNext$ operators)
have been collapsed into single edges.  This ensures that $\eT \setminus
\{root\}$ is an ordered, rooted, regular, binary tree.
The circuits $\circuit{\rho}(\Next)$, $\circuit{\rho}(\WNext)$,
$\circuit{\rho}(\pNext)$, and $\circuit{\rho}(\pWNext)$ are
evaluated one-input-face planar transducer circuits that do not contain any
constants. Hence, any number of these
circuits can be composed resulting in an evaluated one-input-face planar transducer
circuit. 
The composition of planar transducer circuits can be performed in logarithmic
space. The mapping $c$ is defined recursively above. However, it is easy to see
that the whole procedure can be performed iteratively in logarithmic space in
the size of $\rho$ plus the size of $\phi$.
From the above, the first and the second condition for a contraction tree are
clear. The third condition is obtained by applying \prettyref{lem:circuitsound}
to the construction.
\end{proof}

\subsection{Contraction step}\label{sec:contraction-step}

In the following, we describe the contraction of the tree $\eT$. During a
contraction step, a node that is labeled by a constant circuit is merged with
its parent node. The resulting node is contracted into the edge from its sibling
to its grandparent.

\begin{lem}
  Let $\phi_i$ a node of a contraction tree $\eT$ with child nodes $\phi_j$ and
  $\phi_k$ and parent node $p$. 
  Assume $\void{\langle
  \Gamma_{l(\phi_j)},\gamma_{l(\phi_j)},I_{l(\phi_j)},O_{l(\phi_j)}\rangle =}
  \phi_j$ to be a leaf. Let $s$ be the evaluation of $\cir_{\rho}(\phi_j)
  \ccomp l(\langle\phi_i,\phi_j\rangle)$. Let $\eT'=\langle T',t',l' \rangle$, where 
{\small
  \begin{align*}
    &T' =T \setminus \{\phi_j,\phi_i\},\\
    &t' = t \cup \{\langle\phi_p,\phi_k\rangle\} \setminus \{\langle\phi_i,\phi_j\rangle,\langle\phi_i,\phi_k\rangle,\langle\phi_p,\phi_i\rangle\},\\
    &l'(\phi_k) = 
    \begin{cases}
      \text{evaluation of } l(\langle\phi_i,\phi_k\rangle) \ccomp \circuit{\rho}(\cirf{s}(),\phi_i) \ccomp l((\phi_i,p) &\text{if $\phi_j$ is the left child of $\phi_i$},\\
      \text{evaluation of } l(\langle\phi_i,\phi_k\rangle) \ccomp \circuit{\rho}(\phi_i,\cirf{s}()) \ccomp l((\phi_i,p) &\text{if $\phi_j$ is the right child of $\phi_i$},
    \end{cases}\\
    &l'(x) = l(x) \text{ for } x \neq \phi_k.
  \end{align*}
}
  $\eT'$ is a contraction tree and can be computed in $\logDCFL$.
  \label{lem:contraction}
\end{lem}

\begin{proof}
  First, note that by construction of $\eT$ it holds that $\phi_i,\phi_j,\phi_k \neq root$. 
  Clearly, if $\eT\setminus\{root\}$ is an ordered, rooted, regular, binary tree then
  $\eT'\setminus\{root\}$ is an ordered, rooteted, regular, binary tree, as well.
By construction of $\eT'$ a leaf in $\eT'$ is a leaf in $\eT$ as well. Thus,
because $\eT$ is a contraction tree, each leaf in $\eT'$ is an atomic proposition.
Since $\phi_j$ is a leaf $\phi_j$ is an atomic proposition. Due to
\prettyref{lem:circuitsound} $\cir_{\rho}(\phi_j)$ can be composed with
$l(\langle\phi_i,\phi_j\rangle)$ resulting in an one-input-face planar circuit that can be evaluated in
\logDCFL. 
Thus $s$ is a constant circuit of arity $\abs{\rho}$ and
$\circuit{\rho}(\cirf{s}(),\phi_i)$ (respectively $\circuit{\rho}(\phi_i,\cirf{s}())$) is
well defined and of arity $\abs{\rho}$.
Because $\eT$ is a contraction tree and by \prettyref{lem:oifccomp},
$l'(\langle p,\phi_k\rangle)$ is an evaluated one-input-face planar transducer circuit.  
By the definition of $\ccomp$ the input arity of $l'(\langle p,\phi_k\rangle)$
is the input arity of $l(\langle \phi_i,\phi_k \rangle)$ and the output arity of
$l'(\langle p,\phi_k\rangle)$ is the output arity of $l(\langle
p,\phi_i\rangle)$. Because $\eT$ is a contraction tree this arity is
$\abs{\rho}$ in both cases.  All remaining edge labels of $\eT'$ inherit the
arities from $\eT$.
Considering the edge $\langle p,\phi_k\rangle \in t'$, the third condition for a contraction
tree holds, since $\eT$ is a contraction tree, and due to the definition of $\ccomp$, and
because of \prettyref{lem:circuitsound}.  For all other edges, the property is
directly inherited from $\eT$.
The computation of $\eT'$ is in $\logDCFL$ because of \prettyref{lem:oifccomp} and
\prettyref{lem:circuitsound}.
\end{proof}

\subsection{The path checking algorithm.}

Applying \prettyref{lem:init} and \prettyref{lem:contraction} to $\phi$ and $\rho$, we
can use \prettyref{alg:contraction} to obtain an $\ACO(\logDCFL)$ solution 
to the path checking problem. This proves our main theorem:


\begin{proof}[\prettyref{thm:main}] 
Given a \BLTL\ formula $\phi$ and a path $\rho$, convert $\phi$
into positive normal form using only logarithmic space.
A contraction tree $\eT$ is initialized from $\phi$ in logarithmic space by use
of \prettyref{lem:init} and then \prettyref{alg:contraction} is applied to $\eT$
with the contraction step defined in \prettyref{sec:contraction-step}. Note that
the extra $root$ node and the edge $\langle root,\phi\rangle$ in $\eT$ do not
influence the performance of \prettyref{alg:contraction}. The algorithm
terminates when there is only a single leaf node $n$ and a single edge $\langle
root,n\rangle$ left in the contraction tree.  By \prettyref{lem:contraction}, the
contraction algorithm performes in $\ACO(\logDCFL)$.
The value of the first output gate of the evaluation of the circuit $c =
\circuit{\rho}(n) \ccomp l(\langle root,n\rangle)$ is the result. 
By \prettyref{lem:contraction}, $c$ is evaluated and one-input-face and can hence be evaluated in
\logDCFL. 
The whole construction can be executed within $\ACO(\logDCFL)$ in $\abs{\phi} + \abs{\rho}$.
Using \prettyref{lem:bltlprune}, we
can assume that any bound occurring in $\phi$ has at most size $\abs{\rho}$. 
The sum of the bounds is thus polynomial in the size of $\phi$ (without bounds)
and the length of $\rho$. Thus, the overall complexity of $\ACO(\logDCFL)$ is
independent of the encoding of the bounds in $\phi$.
\end{proof}

Since $\LTL$ and $\PLTL$ both are subsets of $\BLTL$
\prettyref{thm:mainltl} and \prettyref{thm:mainpltl} are obtained as corollaries
of \prettyref{thm:main}.


\section{Conclusions}
\label{sec:conclusions}

\noindent We have presented a positive answer to the question whether
\LTL\ can be checked efficiently in parallel on finite paths by giving an 
$\ACO(\logDCFL)$ algorithm for checking formulas of the extended logic \BLTL\
over finite paths. This result is a significant step
forward in the research program towards a complete picture of the
complexities of the path checking problems across the spectrum of
temporal logics, which was started in 2003 by Markey and
Schnoebelen~\cite{Markey+Schnoebelen/03/Model}.  While other
extensions of \LTL, for example with Chop or Past+Now, immediately render
the path checking problem
\poly-complete and, hence,
inherently sequential~\cite{Markey+Schnoebelen/03/Model}, \LTL\ with past and bounds can be checked
efficiently in parallel.

There is a growing practical demand for efficient parallel algorithms,
driven by the increasing availability of powerful (and inherently
parallel) programmable hardware. For example, tools that translate PSL
assertions to hardware-based
monitors~\cite{Dahan+others/05/Combining,BouleZilic2008,FinkbeinerKuhtz09}
can immediately apply our construction to evaluate subformulas
consisting of bounded and past operators in parallel rather than
sequentially.
Similarly, monitoring tools based on \PLTL\ can buffer constant
chunks of the input and then evaluate the buffered input in parallel using
our construction.

The capability of our algorithm to absorb the exponential
succinctness of past and bounds is due to the use of planar circuits
as a representation of partially evaluated subformulas, which allows
the evaluation of the formula to efficiently stop and resume, as
dictated by the dependencies between the subformulas. We expect
that the use of planar circuits as a data structure in parallel
verification algorithms, following the pattern of our construction,
will find applications in other model checking problems as well. 

There are several open questions that deserve further attention.
There is still a gap between $\ACO(\logDCFL)$ and the known
lower bound, $\NCO$. 
There is some hope to further
reduce the upper bound towards $\NCO$, the currently known lower bound, because our construction relies on the algorithm by
Chakraborty and Datta (cf.~\prettyref{thm:ChakrabortyDatta2006}) for
evaluating monotone Boolean planar circuits with all constant gates on
the outer face.  The circuits that appear in our construction actually
exhibit much more structure.  However, we are not aware of any
algorithm that takes advantage of that and performs better than
\logDCFL.
An intriguing question along the way is whether
the path checking complexities of \LTL\ and \BLTL\ are actually the
same: while they are both in \NC, the circuits resulting from \BLTL\ formulas seem to be
combinatorially more complex. Finally, an interesting challenge
is to exploit the apparent ``cheapness'' of the \BLTL\ path checking
problem beyond parallelization, for example in memory-efficient
algorithms.

%
%
%


\bibliographystyle{plain}
\bibliography{literature}

\appendix


\section*{Appendix: Construction of \texorpdfstring{$\circuit{\rho}$}{circuits}}\label{sec:appendix}

\noindent Let $n = \abs{\rho}$.
Let $b \in \Nat$.
Let $I = v_0, \dots, v_{n-1}$ and $O = o_0, \dots, o_{n-1}$. 
Let $G = I \cup O$. 
Let $H=\{g_{i,j} \mid 0 \leq i < n, 0 \leq j \leq b\}$.
Let $O_H = g_{0,0}, \dots, g_{n-1,0}$.
Let $I_H = g_{0,b}, \dots, g_{n-1,b}$.

For an atomic proposition $p$
$\circuit{\rho}(p) = \langle O,l,\varepsilon,O \rangle$,
where $l(o_i) = 1$ iff $p \in \rho_i$ and $\varepsilon$ denotes an empty
input sequence.

\noindent $\circuit{\rho}(\vee,s) = \langle G,l,I,O \rangle$, where 
\begin{align*}
  l(o_i) &= \begin{cases}
    \langle \gaid, v_{i}\rangle &\text{for } s_i = 0 \text{ and } 0 \leq i < n,\\ 
    1 &\text{for } s_i = 1 \text{ and } 0 \leq i < n,
  \end{cases}\\
  l(v_i) &= \gvar \text{ for } 0 \leq i < n;
\end{align*}

\noindent $\circuit{\rho}(s,\vee) = \circuit{\rho}(\vee,s)$;

\noindent $\circuit{\rho}(\wedge,s) = \langle G,l,I,O \rangle$, where
\begin{align*}
  l(o_i) &= \begin{cases}
    \langle \gaid, v_{i}\rangle &\text{for } s_i = 1 \text{ and } 0 \leq i < n,\\ 
    0 &\text{for } s_i = 0 \text{ and } 0 \leq i < n,
  \end{cases}\\
  l(v_i) &= \gvar \text{ for } 0 \leq i < n;
\end{align*}

\noindent $\circuit{\rho}(s,\wedge) = \circuit{\rho}(\wedge,s)$;


\noindent $\circuit{\rho}(\Next) = \langle G,l,I,O \rangle$, where 
\begin{align*}
  l(o_i) &= \begin{cases}
    \langle \gaid, v_{i+1}\rangle &\text{for } 0 \leq i < n-1,\\ 
    0 &\text{for } i = n-1, 
  \end{cases}\\
  l(v_i) &= \gvar \text{ for } 0 \leq i < n;
\end{align*}

\noindent $\circuit{\rho}(\WNext) = \langle G,l,I,O \rangle$, where 
\begin{align*}
  l(o_i) &= \begin{cases}
    \langle \gaid, v_{i+1}\rangle &\text{for } 0 \leq i < n-1,\\ 
    1 &\text{for } i = n-1, 
  \end{cases}\\
  l(v_i) &= \gvar \text{ for } 0 \leq i < n;
\end{align*}

\noindent $\circuit{\rho}(\until,s) = \langle G,l,I,O \rangle$, where 
\begin{align*}
  l(o_{i}) &= \begin{cases}
    \langle \gand, v_i, o_{i+1}\rangle &\text{for } 0 \leq i < n-1 \text{ and } s_i = 0,\\ 
    1 &\text{for } 0 \leq i < n-1 \text{ and } s_i = 1,\\
    s_{n-1} &\text{for } i = n-1,
  \end{cases}\\
  l(v_i) &= \gvar \text{ for } 0 \leq i < n;
\end{align*}

\noindent $\circuit{\rho}(s,\until) = \langle G,l,I,O \rangle$, where 
\begin{align*}
  l(o_i) &= \begin{cases}
    \langle \gor, v_i, o_{i+1}\rangle &\text{for } 0 \leq i < n-1 \text{ and } s_i = 1,\\
    \langle \gaid, v_i \rangle &\text{for } 0 \leq i < n-1 \text{ and } s_i = 0,\\
    \langle \gaid, v_{n-1} \rangle &\text{for } i = n-1,
  \end{cases}\\
  l(v_i) &= \gvar \text{ for } 0 \leq i < n;
\end{align*}

\noindent $\circuit{\rho}(\release,s) = \langle G,l,I,O \rangle$, where 
\begin{align*}
  l(o_{i}) &= \begin{cases}
    \langle \gor, v_i, o_{i+1}\rangle &\text{for } 0 \leq i < n-1 \text{ and } s_i = 1,\\ 
    0 &\text{for } 0 \leq i < n-1 \text{ and } s_i = 0,\\
    s_{n-1} &\text{for } i = n-1,
  \end{cases}\\
  l(v_i) &= \gvar \text{ for } 0 \leq i < n;
\end{align*}

\noindent $\circuit{\rho}(s,\release) = \langle G,l,I,O \rangle$, where 
\begin{align*}
  l(o_i) &= \begin{cases}
    \langle \gand, v_i, o_{i+1}\rangle &\text{for } 0 \leq i < n-1 \text{ and } s_i = 0,\\
    \langle \gaid, v_i \rangle &\text{for } 0 \leq i < n-1 \text{ and } s_i = 1,\\
    \langle \gaid, v_{n-1} \rangle &\text{for } i = n-1,
  \end{cases}\\
  l(g) &= \gvar \text{ for } g \in I;
\end{align*}

\noindent $\circuit{\rho}(\buntil{b},s) = \langle G,l,I,O \rangle$, where
\begin{align*}
  l(o_{i}) &= \begin{cases}
    1 &\text{for } 0 \leq i < n \text{ and } s_i = 1,\\
    0 &\text{for } 0 \leq i < n \text{ and } \forall j, i \leq j < \min(i+b,n) \qdot s_j = 0,\\
    \langle \gand, v_i, o_{i+1}\rangle &\text{for } 0 \leq i < n-1 \text{ and } s_i = 0 \text{ and } \exists j,i < j < \min(i+b,n) \qdot s_j = 1,
  \end{cases}\\
  l(v_i) &= \gvar \text{ for } 0 \leq i < n;
\end{align*}

\noindent $\circuit{\rho}(s,\buntil{b}) = \langle H,l,I_H,O_H \rangle$, where
\begin{equation*}
  l(g_{i,j}) = \begin{cases}
    \left\langle \gaid,g_{i,j+1} \right\rangle &\text{for } 0 \leq i < n-1 \text{ and } j < b \text{ and } s_i = 0,\\
    \left\langle \gor,g_{i,j+1},g_{i+1,j+1} \right\rangle &\text{for } 0 \leq i < n-1 \text{ and } j < b \text{ and } s_i = 1,\\
    \left\langle \gaid,g_{i,j+1} \right\rangle &\text{for } i = n-1 \text{ and } j < b,\\
    \gvar &\text{for } j = b,
  \end{cases}
\end{equation*}

\noindent $\circuit{\rho}(\brelease{b},s) = \langle G,l,I,O \rangle$, where
\begin{align*}
  l(o_{i}) &= \begin{cases}
    0 &\text{for } 0 \leq i < n \text{ and } s_i = 0,\\
    1 &\text{for } 0 \leq i < n \text{ and } \forall j, i \leq j < \min(i+b,n) \qdot s_j = 1,\\
    \langle \gor, v_i, o_{i+1}\rangle &\text{for } 0 \leq i < n-1 \text{ and } s_i = 1 \text{ and } \exists j,i < j < \min(i+b,n) \qdot s_j = 0,
  \end{cases}\\
  l(v_i) &= \gvar \text{ for } 0 \leq i < n;
\end{align*}

\noindent $\circuit{\rho}(s,\brelease{b}) = \langle H,l,I_H,O_H \rangle$, where 
\begin{equation*}
  l(g_{i,j}) = \begin{cases}
    \left\langle \gaid,g_{i,j+1} \right\rangle &\text{for } 0 \leq i < n-1 \text{ and } j < b \text{ and } s_i = 1,\\
    \left\langle \gand,g_{i,j+1},g_{i+1,j+1} \right\rangle &\text{for } 0 \leq i < n-1 \text{ and } j < b \text{ and } s_i = 0,\\
    \left\langle \gaid,g_{i,j+1} \right\rangle &\text{for } i = n-1 \text{ and } j < b,\\
    \gvar &\text{for } j = b;
  \end{cases}
\end{equation*}


\noindent $\circuit{\rho}(\pNext) = \langle G,l,I,O \rangle$, where 
\begin{align*}
  l(o_i) &= \begin{cases}
    \langle \gaid, v_{i-1}\rangle &\text{for } 0 < i < n,\\ 
    0 &\text{for } i = 0, 
  \end{cases}\\
  l(v_i) &= \gvar \text{ for } 0 \leq i < n;
\end{align*}

\noindent $\circuit{\rho}(\pWNext) = \langle G,l,I,O \rangle$, where 
\begin{align*}
  l(o_i) &= \begin{cases}
    \langle \gaid, v_{i-1}\rangle &\text{for } 0 < i < n,\\ 
    1 &\text{for } i = 0, 
  \end{cases}\\
  l(v_i) &= \gvar \text{ for } 0 \leq i < n;
\end{align*}

\noindent $\circuit{\rho}(\puntil,s) = \langle G,l,I,O \rangle$, where 
\begin{align*}
  l(o_{i}) &= \begin{cases}
    \langle \gand, v_i, o_{i-1}\rangle &\text{for } 0 < i < n \text{ and } s_i = 0,\\ 
    1 &\text{for } 0 < i < n \text{ and } s_i = 1,\\
    s_0 &\text{for } i = 0,
  \end{cases}\\
  l(v_i) &= \gvar \text{ for } 0 \leq i < n;
\end{align*}

\noindent $\circuit{\rho}(s,\puntil) = \langle G,l,I,O \rangle$, where 
\begin{align*}
  l(o_i) &= \begin{cases}
    \langle \gor, v_i, o_{i-1}\rangle &\text{for } 0 < i < n \text{ and } s_i = 1,\\
    \langle \gaid, v_i \rangle &\text{for } 0 < i < n \text{ and } s_i = 0,\\
    \langle \gaid, v_{0} \rangle &\text{for } i = 0,
  \end{cases}\\
  l(v_i) &= \gvar \text{ for } 0 \leq i < n;
\end{align*}

\noindent $\circuit{\rho}(\prelease,s) = \langle G,l,I,O \rangle$, where 
\begin{align*}
  l(o_{i}) &= \begin{cases}
    \langle \gor, v_i, o_{i-1}\rangle &\text{for } 0 < i < n \text{ and } s_i = 1,\\ 
    0 &\text{for } 0 < i < n \text{ and } s_i = 0,\\
    s_{0} &\text{for } i = 0,
  \end{cases}\\
  l(v_i) &= \gvar \text{ for } 0 \leq i < n;
\end{align*}

\noindent $\circuit{\rho}(s,\prelease) = \langle G,l,I,O \rangle$, where 
\begin{align*}
  l(o_i) &= \begin{cases}
    \langle \gand, v_i, o_{i-1}\rangle &\text{for } 0 < i < n \text{ and } s_i = 0,\\
    \langle \gaid, v_i \rangle &\text{for } 0 < i < n \text{ and } s_i = 1,\\
    \langle \gaid, v_{0} \rangle &\text{for } i = 0,
  \end{cases}\\
  l(g) &= \gvar \text{ for } g \in I;
\end{align*}

\noindent $\circuit{\rho}(\pbuntil{b},s) = \langle G,l,I,O \rangle$, where
\begin{align*}
  l(o_{i}) &= \begin{cases}
    1 &\text{for } 0 \leq i < n \text{ and } s_i = 1,\\
    0 &\text{for } 0 \leq i < n \text{ and } \forall j, i \geq j > \max(i-b,-1) \qdot s_j = 0,\\
    \langle \gand, v_i, o_{i-1}\rangle &\text{for } 0 < i < n \text{ and } s_i = 0 \text{ and } \exists j,i > j >\max(i-b,-1) \qdot s_j = 1,
  \end{cases}\\
  l(v_i) &= \gvar \text{ for } 0 \leq i < n;
\end{align*}

\noindent $\circuit{\rho}(s,\pbuntil{b}) = \langle H,l,I_H,O_H \rangle$, where
\begin{equation*}
  l(g_{i,j}) = \begin{cases}
    \left\langle \gaid,g_{i,j+1} \right\rangle &\text{for } 0 < i < n \text{ and } j < b \text{ and } s_i = 0,\\
    \left\langle \gor,g_{i,j+1},g_{i-1,j+1} \right\rangle &\text{for } 0 < i < n \text{ and } j < b \text{ and } s_i = 1,\\
    \left\langle \gaid,g_{0,j+1} \right\rangle &\text{for } j < b,\\
    \gvar &\text{for } j = b;
  \end{cases}
\end{equation*}

\noindent $\circuit{\rho}(\pbrelease{b},s) = \langle G,l,I,O \rangle$, where
\begin{align*}
  l(o_{i}) &= \begin{cases}
    0 &\text{for } 0 \leq i < n \text{ and } s_i = 0,\\
    1 &\text{for } 0 \leq i < n \text{ and } \forall j, i \geq j > \min(i-b,-1) \qdot s_j = 1,\\
    \langle \gor, v_i, o_{i-1}\rangle &\text{for } 0 < i < n \text{ and } s_i = 1 \text{ and } \exists j,i > j > \max(i-b,-1) \qdot s_j = 0,
  \end{cases}\\
  l(v_i) &= \gvar \text{ for } 0 \leq i < n;
\end{align*}

\noindent $\circuit{\rho}(s,\pbrelease{b}) = \langle H,l,I_H,O_H \rangle$, where 
\begin{equation*}
  l(g_{i,j}) = \begin{cases}
    \left\langle \gaid,g_{i,j+1} \right\rangle &\text{for } 0 < i < n \text{ and } j < b \text{ and } s_i = 1,\\
    \left\langle \gand,g_{i,j+1},g_{i-1,j+1} \right\rangle &\text{for } 0 < i < n \text{ and } j < b \text{ and } s_i = 0,\\
    \left\langle \gaid,g_{0,j+1} \right\rangle &\text{for } j < b,\\
    \gvar &\text{for } j = b.
  \end{cases}
\end{equation*}


\end{document}

%% file: macros.tex
\newcommand{\qdot}{.}

\newcommand{\T}{\ensuremath \mathcal{T}}

\newcommand{\void}[1]{}

\newcommand{\Epsilon}{\mathsf{E}}
\DeclareMathOperator{\udot}{\dot\cup}
\DeclareMathOperator{\ccomp}{\circ}
\newcommand{\ccompf}[1]{\circ_{\!f}}
\newcommand{\circuit}[1]{\ensuremath \mathrm{cir}_{{#1}}}
\newcommand{\eT}{\mathcal{T}}

\newcommand{\parent}{\mathrm{parent}}
\newcommand{\cirf}[1]{\ensuremath f_{#1}}

\newcommand{\LTL}{LTL}
\newcommand{\PLTL}{LTL+Past}
\newcommand{\BLTL}{BLTL+Past}

\newcommand{\Bool}{\mathbb{B}}
\newcommand{\Nat}{\mathbb{N}}

\newcommand\zero{{\tt 0}}
\newcommand\one{{\tt 1}}

\DeclareMathOperator{\release}{\mathrm{R}}
\DeclareMathOperator{\until}{\mathrm{U}}
\newcommand{\buntil}[1]{\,\mathrm{U}_{#1}\,}

\newcommand{\brelease}[1]{\,\mathrm{R}_{#1}\,}
\newcommand{\Next}{\mathrm{X}^{\exists}}
\newcommand{\WNext}{\mathrm{X}^{\forall}}

\newcommand{\G}{\,\mathrm{G}\,}

\DeclareMathOperator{\prelease}{\mathrm{T}} 
\DeclareMathOperator{\puntil}{\mathrm{S}}   
\newcommand{\pbuntil}[1]{\,\mathrm{S}_{#1}\,} 
\newcommand{\pbreleaseNoblank}[1]{\,\mathrm{T}_{#1}}
\newcommand{\pbrelease}[1]{\,\mathrm{T}_{#1}\,}
\newcommand{\pNext}{\mathrm{Y}^{\exists}} 
\newcommand{\pWNext}{\mathrm{Y}^{\forall}}

\newcommand{\clx}[1]{{\sf #1}}
\newcommand{\logDCFL}{\mbox{\clx{logDCFL}}}
\newcommand{\logCFL}{\mbox{\clx{logCFL}}}
\newcommand{\NL}{\mbox{\clx{NL}}}
\newcommand{\logspace}{\mbox{\clx{L}}}
\newcommand{\poly}{\mbox{\clx{P}}}
\newcommand{\AC}{\mbox{\clx{AC}}}
\newcommand{\ACO}{\mbox{\clx{AC}}^1}

\newcommand{\NC}{\mbox{\clx{NC}}}
\newcommand{\NCO}{\mbox{\clx{NC}}^1}

\newcommand\gop{\ensuremath \mathit{op}}
\newcommand\gleft{\ensuremath \mathit{left}}
\newcommand\gright{\ensuremath \mathit{right}}
\newcommand\gaid{\ensuremath \mathit{id}}
\newcommand\gsuc{\ensuremath \mathit{suc}}
\newcommand\gand{\ensuremath \mathit{and}}
\newcommand\gor{\ensuremath \mathit{or}}
\newcommand\gvar{\ensuremath \mathit{?}}

\newcommand{\abs}[1]{\| #1 \|}
\newcommand{\ceil}[1]{\left\lceil#1 \right\rceil}

\DeclareMathOperator{\ddep}{\cdot\hspace{-1ex}\succ}

\newcommand{\cir}{\ensuremath \mbox{cir}}
\newcommand{\gr}{\ensuremath \mbox{gr}}
\newcommand{\constants}{\ensuremath \mbox{const}}
\newcommand{\inputs}{\ensuremath \mbox{src}}
\newcommand{\vars}{\ensuremath \mbox{var}}

%% file: figure-simple.tex
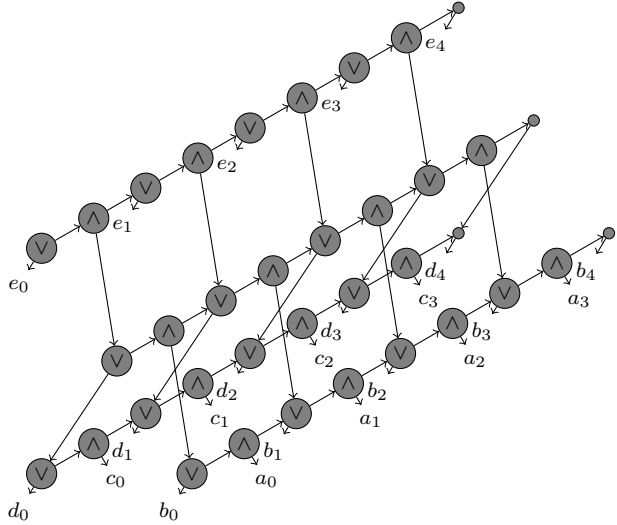
\begin{figure}
  \begin{center}

    \newcommand{\mj}{7}
    \newcommand{\gate}[2]{g{#1-#2}}   
    \newcommand{\gatee}[2]{e{#1-#2}}  
    \newcommand{\gateo}[2]{o{#1-#2}}  
    \newcommand{\gatephi}[2]{phi#1-#2}

    \newcommand{\andgate}{}
    \newcommand{\orgate}{}
    \newcommand{\idgate}[3][]{$[c#1_{#2,#3}]$}


  \def\mj{4}
  \tikzstyle{perspective dimetric}=[%
  x={(canvas polar cs:angle=0,radius=1cm)}, 
  y={(canvas polar cs:angle=90,radius=1cm)}, 
  z={(canvas polar cs:angle=30,radius=0.8cm)} 
]
  \begin{tikzpicture}[join=round, 
    perspective dimetric, font=\scriptsize,
    <-
    ]

     \tikzstyle{nn}=[circle,inner sep=1.5pt]
     \tikzstyle{an}=[nn,draw,fill=gray];
     \tikzstyle{on}=[nn,draw,fill=gray];
     \tikzstyle{cn}=[nn];

     \def\nn#1,#2,#3{#1x#2x#3}
     \def\constant(#1,#2,#3,#4,#5,#6){%

       \def\xxx{1} \ifnum \mj = #4 \if r#6 \def\xxx{0} \fi\fi
       \if 1\xxx
       \if r#6 \pgfmathsetmacro{\zz}{1} \else \pgfmathsetmacro{\zz}{0} \fi

       \node[cn] (\nn#1,0,#4) at ($(#2,#3,\zz)+2*(0,0,#4)$) {$#5_{#4}$};

       \if l#6 \def\d{0} \else \def\d{1} \fi
       
       \def\xxx{1} \ifnum #4 = \mj \ifnum \d = 1 \def\xxx{0} \fi \fi
       \if 1\xxx
       \pgfmathtruncatemacro{\epr}{#1 / 10}
       \ifnum \epr < 1 \def\pr{1} \else \def\pr{\epr} \fi
       \if 0 #3 \else \draw (\nn#1,0,#4) -- (\nn\pr,\d,#4); \fi
       \fi
       \fi
     }

     \def\until(#1,#2,#3,#4,#5){%
       \ifnum #4 = \mj \def\lab{} \else \def\lab{$\vee$} \fi
       \node[on] (\nn#1,0,#4) at ($(#2,#3,0)+2*(0,0,#4)$) {\lab};
      
       \ifnum #4 < \mj
       \node[an] (\nn#1,1,#4) at ($(#2,#3,1)+2*(0,0,#4)$) {$\wedge$};
       \draw[->] (\nn#1,0,#4) -- (\nn#1,1,#4); \fi

      
       \pgfmathtruncatemacro{\pz}{#4+1}
       \ifnum #4 < \mj \draw (\nn#1,0,\pz) -- (\nn#1,1,#4); \fi
       
       \if l#5 \def\d{0} \else \def\d{1} \fi
       
       \def\xxx{1} \ifnum #4 = \mj \ifnum \d = 1 \def\xxx{0} \fi \fi
       \if 1\xxx
       \pgfmathtruncatemacro{\epr}{#1 / 10}
       \ifnum \epr < 1 \def\pr{1} \else \def\pr{\epr} \fi
       \if 0#3 \else \draw (\nn#1,0,#4) -- (\nn\pr,\d,#4); \fi
       \fi
     }

     \foreach \z in {\mj,3,2,1,0} {
       \begin{scope}[transparent]
       \until(1,0,0,\z,x)
       \constant(10,-.3,-.5,\z,e,l)
       \until(11,1,-1.5,\z,r)
       \until(110,0,-3,\z,l)
       \until(111,2,-3,\z,r)
       \constant(1100,-.3,-3.5,\z,d,l)
       \constant(1101,0.3,-3.5,\z,c,r)
       \constant(1110,1.7,-3.5,\z,b,l)
       \constant(1111,2.3,-3.5,\z,a,r)
     \end{scope}
       \constant(1100,-.3,-3.5,\z,d,l)
       \constant(1101,0.3,-3.5,\z,c,r)
       \constant(1110,1.7,-3.5,\z,b,l)
       \constant(1111,2.3,-3.5,\z,a,r)
       \until(111,2,-3,\z,r)
       \until(110,0,-3,\z,l)
       \until(11,1,-1.5,\z,r)
       \constant(10,-.3,-.5,\z,e,l)
       \until(1,0,0,\z,x)
     }
      
  \end{tikzpicture}
  \caption{Circuit resulting from unrolling the LTL formula $\left(\left(
  a\until b \right) \until \left( c \until d \right)\right) \until e$ over a
  path $\rho$ of length 5. We denote the value of an atomic proposition $p$ at
  a path position $i=0,\dots,4$ by $p_i$.%
  } \label{fig:simple}
\end{center}
\end{figure}

%% file: figure-simple-non-planar.tex
\begin{figure}
  \begin{center}

    \newcommand{\mj}{7}
    \newcommand{\gate}[2]{g{#1-#2}}   
    \newcommand{\gatee}[2]{e{#1-#2}}  
    \newcommand{\gateo}[2]{o{#1-#2}}  
    \newcommand{\gatephi}[2]{phi#1-#2}



  \def\mj{4}
  \tikzstyle{perspective dimetric}=[%
  x={(canvas polar cs:angle=0,radius=1cm)}, 
  y={(canvas polar cs:angle=90,radius=1cm)}, 
  z={(canvas polar cs:angle=30,radius=0.8cm)} 
]
  \begin{tikzpicture}[join=round, 
    perspective dimetric, font=\scriptsize,
    <-
    ]

     \tikzstyle{nn}=[circle,inner sep=1.5pt]
     \tikzstyle{an}=[nn,draw,fill=gray];
     \tikzstyle{on}=[nn,draw,fill=gray];
     \tikzstyle{cn}=[nn];

     \def\nn#1,#2,#3{#1x#2x#3}
     \def\constant(#1,#2,#3,#4,#5,#6){%

       \def\xxx{1} \ifnum \mj = #4 \if r#6 \def\xxx{0} \fi\fi
       \if 1\xxx
       \if r#6 \pgfmathsetmacro{\zz}{1} \else \pgfmathsetmacro{\zz}{0} \fi

       \node[cn] (\nn#1,0,#4) at ($(#2,#3,\zz)+2*(0,0,#4)$) {$#5_{#4}$};

       \if l#6 \def\d{0} \else \def\d{1} \fi
       
       \def\xxx{1} \ifnum #4 = \mj \ifnum \d = 1 \def\xxx{0} \fi \fi
       \if 1\xxx
       \pgfmathtruncatemacro{\epr}{#1 / 10}
       \ifnum \epr < 1 \def\pr{1} \else \def\pr{\epr} \fi
       \if 0 #3 \else \draw (\nn#1,0,#4) -- (\nn\pr,\d,#4); \fi
       \fi
       \fi
     }

     \def\until(#1,#2,#3,#4,#5){%
       \ifnum #4 = \mj \def\lab{} \else \def\lab{$\vee$} \fi
       \node[on] (\nn#1,0,#4) at ($(#2,#3,0)+2*(0,0,#4)$) {\lab};
      
       \ifnum #4 < \mj
       \node[an] (\nn#1,1,#4) at ($(#2,#3,1)+2*(0,0,#4)$) {$\wedge$};
       \draw[->] (\nn#1,0,#4) -- (\nn#1,1,#4); \fi

       \pgfmathtruncatemacro{\pz}{#4+1}
       \ifnum #4 < \mj \draw (\nn#1,0,\pz) -- (\nn#1,1,#4); \fi
       
       \if l#5 \def\d{0} \else \def\d{1} \fi
       
       \def\xxx{1} \ifnum #4 = \mj \ifnum \d = 1 \def\xxx{0} \fi \fi
       \if 1\xxx
       \pgfmathtruncatemacro{\epr}{#1 / 10}
       \ifnum \epr < 1 \def\pr{1} \else \def\pr{\epr} \fi
       \if 0#3 \else \draw (\nn#1,0,#4) -- (\nn\pr,\d,#4); \fi
       \fi
     }

     \foreach \z in {\mj,3,2,1,0} {
       \begin{scope}[transparent]
       \until(1,0,0,\z,x)
       \constant(10,-.3,-.5,\z,e,l)
       \until(11,1,-1.5,\z,r)
       \until(110,0,-3,\z,l)
       \until(111,2,-3,\z,r)
       \constant(1100,-.3,-3.5,\z,d,l)
       \constant(1101,0.3,-3.5,\z,c,r)
       \constant(1110,1.7,-3.5,\z,b,l)
       \constant(1111,2.3,-3.5,\z,a,r)
     \end{scope}
       \constant(1100,-.3,-3.5,\z,d,l)
       \constant(1101,0.3,-3.5,\z,c,r)
       \constant(1110,1.7,-3.5,\z,b,l)
       \constant(1111,2.3,-3.5,\z,a,r)
       \until(111,2,-3,\z,r)
       \until(110,0,-3,\z,l)
       \until(11,1,-1.5,\z,r)
       \constant(10,-.3,-.5,\z,e,l)
       \until(1,0,0,\z,x)
     }

    \tikzstyle{group}=[-,draw=red,line width=3pt,rounded corners=2pt, draw opacity=0.7]
    \tikzstyle{group node}=[group,fill, color=red, fill opacity=0.6]
    \draw[group node] (\nn1,1,1.west) -- (\nn1,1,1.north)
      -- (\nn1,1,3.north) -- (\nn1,1,3.east)
      -- (\nn1,1,3.south) -- (\nn1,1,1.south) 
      -- cycle ;

    \draw[group node] (\nn110,0,1.west) -- (\nn110,0,1.north)
      -- (\nn110,0,3.north) -- (\nn110,0,3.east)
      -- (\nn110,0,3.south) -- (\nn110,0,1.south) 
      -- cycle ;

    \draw[group node] (\nn11,0,1.west) -- (\nn11,0,1.north) -- (\nn11,0,1.east) -- (\nn11,0,1.south) -- cycle;
    \draw[group node] (\nn11,0,2.west) -- (\nn11,0,2.north) -- (\nn11,0,2.east) -- (\nn11,0,2.south) -- cycle;
    \draw[group node] (\nn11,0,3.west) -- (\nn11,0,3.north) -- (\nn11,0,3.east) -- (\nn11,0,3.south) -- cycle;

    \draw[group] (\nn1,1,1) -- (\nn11,0,1);
    \draw[group] (\nn1,1,2) -- (\nn11,0,2);
    \draw[group] (\nn1,1,3) -- (\nn11,0,3);
    \draw[group] (\nn11,0,1) -- (\nn110,0,1);
    \draw[group] (\nn11,0,2) -- (\nn110,0,2);
    \draw[group] (\nn11,0,3) -- (\nn110,0,3);
    \draw[group] (\nn11,0,1) -- (\nn11,0,2);
    \draw[group] (\nn11,0,2) -- (\nn11,0,3);
    \draw[group] (\nn110,0,1) -- (\nn110,0,0.center) -- (\nn11,0,0.center) -- (\nn1,1,0.center) -- (\nn1,1,1);
    \draw[group] (\nn11,0,1) -- (\nn11,1,0.center) -- (\nn111,0,0.center) -- (\nn111,0,3.center) -- (\nn11,1,3.center) -- (\nn11,0,3);


  \end{tikzpicture}
  \caption{Circuit resulting from unrolling the LTL formula $\left(\left(
  a\until b \right) \until \left( c \until d \right)\right) \until e$ over a
  path $\rho$ of length 5. The red colored minor of the graph of the circuit is a
  $K_5$.  Thus the circuit is not planar.} \label{fig:simple-non-planar}
\end{center}
\end{figure}
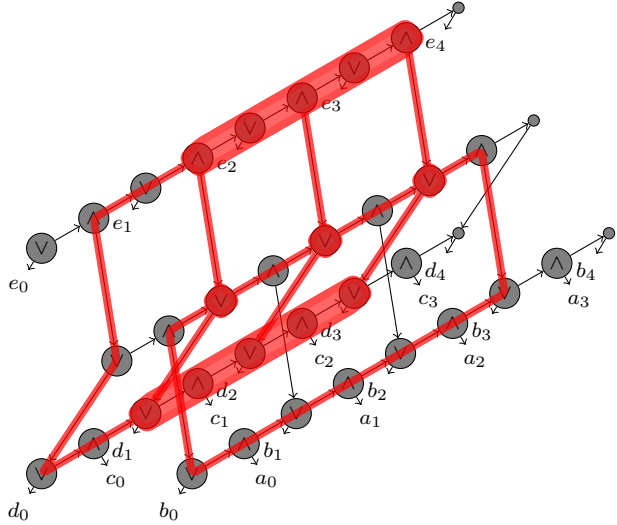

%% file: figure-non-planar.tex
\begin{figure}[t]
  \begin{center}

    \newcommand{\mj}{7}
    \newcommand{\mk}{3}
    \newcommand{\gate}[2]{j{#1}k{#2}}
    \newcommand{\gatephi}[2]{phi{#1}x{#2}}
    \newcommand{\gatepsi}[1]{psi#1}

    \newcommand{\andgate}[2]{$\wedge$}
    \newcommand{\orgate}[2]{$\vee$}
    \newcommand{\idgate}[2]{id}

    \scalebox{1}{%
  \begin{tikzpicture}[font=\small,scale=1.8,<-]

    \foreach \k in {\mk,...,0} {
    \pgfmathtruncatemacro{\start}{(\mk-\k)}

    \foreach \j in {0,...,\mj} {

    \ifthenelse { \k=0 \OR \j=\mj} 
    {\node (\gate{\j}{\k}) at ($(\j,\k)$) {\idgate{\j}{\k}};}
    {\node (\gate{\j}{\k}) at ($(\j,\k)$) {\orgate{\j}{\k}};}
    \ifnum \k=0 
      \node[color=gray] (\gatepsi{\j}) at (\j,-0.5) {$\psi_{\j}$};
      \draw (\gatepsi{\j}) -- (\gate{\j}{\k});
    \fi

    \ifnum \k=\mk \ifnum \j>0
    \pgfmathtruncatemacro{\yyy}{\k+1}
    \pgfmathtruncatemacro{\xxx}{\j-1}
    \node[color=gray] (phi\xxx) at ($(\xxx,-.5) - (0,.5)$) {$\phi_{\xxx}$};
    \fi\fi

    \ifnum \k<\mk
    \pgfmathtruncatemacro{\yyy}{\k+1}
      \draw (\gate{\j}{\k}) -- (\gate{\j}{\yyy});
      \ifnum \j>0
        \pgfmathtruncatemacro{\xxx}{\j-1}
        \node (\gatephi{\xxx}{\yyy}) at ($(\gate{\xxx}{\yyy})!0.5!(\gate{\j}{\k})$) {\andgate{\xxx}{\yyy}};
        \draw[<-] (\gate{\j}{\k}) -- (\gatephi{\xxx}{\yyy});
        \draw[<-] (\gatephi{\xxx}{\yyy}) -- (\gate{\xxx}{\yyy});

         \draw[->] (\gatephi{\xxx}{\yyy}) edge[bend left] (phi\xxx);
      \fi
    \fi
    }}

    \tikzstyle{k33}=[-,color=red,draw opacity=.7,line width=3pt,draw,circle,rounded corners=2pt]
    \coordinate (k0) at (\gate{0}{3});
    \node[k33] (k1) at (\gate{0}{2}) {};
    \coordinate (k2) at (\gate{0}{1});

    \node[k33] (k3) at (\gatephi{0}{3}) {};
    \node[k33] (k4) at (\gatephi{0}{2}) {};
    \node[k33] (k5) at (\gatephi{0}{1}) {};

    \coordinate (k6) at (\gate{1}{2});
    \node[k33] (k7) at (\gate{1}{1}) {};
    \coordinate (k8) at (\gate{1}{0});

    \node[k33] (k9) at (phi0) {};

    \draw[k33] (k1) -- (k0) -- (k3);
    \draw[k33] (k1) -- (k4);
    \draw[k33] (k1) -- (k2) -- (k5);
    \draw[k33] (k7) -- (k6) -- (k3);
    \draw[k33] (k7) -- (k4);
    \draw[k33] (k7) -- (k8) -- (k5);
    \draw[k33] (k3) edge[bend left] (k9);
    \draw[k33] (k4) edge[bend left] (k9);
    \draw[k33] (k5) edge[bend left] (k9);

  \end{tikzpicture}}
  \caption{The circuit for the bounded formula $\phi \buntil{3} \psi$. Since the red colored
  subgraph is a $K_{3,3}$, the circuit has no planar embedding. However, if the
  $\phi_i$-gates are constants, then propagating the constants 
  eliminates the edges that prevent the shown embedding from being planar. 
} \label{fig:non-planar}
\end{center}
\end{figure}
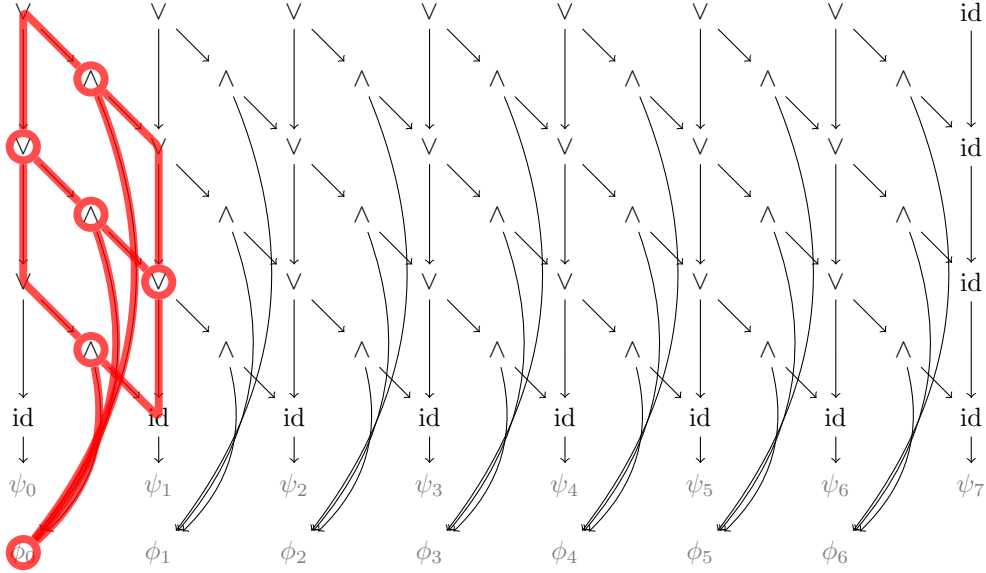%

%% file: figure-planar.tex
\begin{figure}[t]
  \begin{center}

    \newcommand{\mj}{7}
    \newcommand{\mk}{3}
    \newcommand{\gate}[2]{j{#1}k{#2}}
    \newcommand{\gatephi}[2]{phi{#1}x{#2}}
    \newcommand{\gatepsi}[1]{psi#1}

    \newcommand{\andgate}[2]{$\wedge$}
    \newcommand{\orgate}[2]{$\vee$}
    \newcommand{\idgate}[2]{id}

    \scalebox{1}{%
  \begin{tikzpicture}[font=\small,scale=1.8]

    \foreach \k in {\mk,...,0} {
    \pgfmathtruncatemacro{\start}{(\mk-\k)}
   
    \foreach \j in {0,...,\mj} {

    \ifthenelse { \k=0 \OR \j=\mj} 
    {\node (\gate{\j}{\k}) at ($(\j,\k)$) {\idgate{\j}{\k}};}
    {\node (\gate{\j}{\k}) at ($(\j,\k)$) {\orgate{\j}{\k}};}
    \ifnum \k=0 
      \node[color=gray] (\gatepsi{\j}) at (\j,-0.5) {$\psi_{\j}$};
      \draw[<-] (\gatepsi{\j}) -- (\gate{\j}{\k});
    \fi


    \ifnum \k<\mk
    \pgfmathtruncatemacro{\yyy}{\k+1}
      \draw[<-] (\gate{\j}{\k}) -- (\gate{\j}{\yyy});
      \ifnum \j>0
        \pgfmathtruncatemacro{\xxx}{\j-1}
        \node (\gatephi{\xxx}{\yyy}) at ($(\gate{\xxx}{\yyy})!0.5!(\gate{\j}{\k})$) {\andgate{\xxx}{\yyy}};
        \draw[<-] (\gate{\j}{\k}) -- (\gatephi{\xxx}{\yyy});
        \draw[<-](\gatephi{\xxx}{\yyy}) -- (\gate{\xxx}{\yyy});
        \draw[<-] (\gatephi{\xxx}{\yyy}) -- (\gate{\xxx}{\yyy});
        \draw[->] (\gatephi{\xxx}{\yyy}) -- +(0,-0.2) node[anchor=north,gray] (phiin\xxx\yyy) {$\phi_\xxx$};
      \fi
    \fi
    }}

    \foreach \j/\v in {0/0,1/0,2/0,3/0,4/1,5/0,6/1,7/0} {
      \node[color=gray,font=\large, ultra thick] at (\gatepsi{\j}.south) {$\mathbf\v$};
    }
    \tikzstyle{group}=[draw=gray,dashed,line width=3pt,rounded corners=15pt, draw opacity=0.7]

    \draw[group] ($(\gatephi{4}{1})+(-.25,0)$) -- ($(\gatephi{4}{1})+(.25,.25)$)
              -- ($(\gate{5}{0})+(.25,0)$) -- ($(\gate{5}{0})+(-.25,-.25)$) 
              -- cycle;
    \draw[group] ($(\gate{4}{3})+(-.25,.25)$) -- ($(\gate{4}{3})+(.25,.25)$) 
              -- ($(\gate{4}{0})+(.25,-.25)$) -- ($(\gate{4}{0})+(-.25,-.25)$) 
              -- cycle;
    \draw[group] ($(\gate{6}{3})+(-.25,.25)$) -- ($(\gate{6}{3})+(.25,.25)$)
              -- ($(\gate{6}{1})+(.25,0)$) -- ($(\gatephi{6}{1})+(-.25,0)$)
              -- ($(\gate{6}{0})+(.25,-.25)$) -- ($(\gate{6}{0})+(-.25,-.25)$) 
              -- cycle;
    \draw[group] ($(\gatephi{6}{3})+(-.25,.25)$) 
              -- ($(\gate{7}{3})+(-.25,.25)$) -- ($(\gate{7}{3})+(.25,.25)$)
              -- ($(\gate{7}{0})+(.25,-.25)$) -- ($(\gate{7}{0})+(-.25,-.25)$) 
              -- ($(\gatephi{6}{1})+(-.25,0)$) 
              -- cycle;
    \draw[group] ($(\gate{0}{3})+(-.25,.25)$) -- ($(\gate{0}{3})+(.25,.25)$) 
              -- ($(\gate{3}{0})+(.25,.25)$) -- ($(\gate{3}{0})+(.25,-.25)$) 
              -- ($(\gate{0}{0})+(-.25,-.25)$) 
              -- cycle;

    \tikzstyle{grouplabel}=[gray,opacity=0.9]          
    \node[grouplabel] at (1,1.4) {\Huge\bf 0};
    \node[grouplabel] at (4,1.4) {\Huge\bf 1};
    \node[grouplabel] at (6,1.4) {\Huge\bf 1};
    \node[grouplabel] at (7,1.4) {\Huge\bf 0};
    \node[grouplabel] at (4.8,0.3) {\Huge\bf 0};

    \tikzstyle{result}=[cap=round,join=round,color=gray,font=\large, line width=3pt, draw opacity=0.7]
    \draw[result] (\gate{1}{3}) -- (\gatephi{1}{3});
    \draw[result] (\gatephi{1}{3}) -- (phiin13);
    \draw[result] (\gatephi{1}{3}) -- (\gate{2}{2});
    \draw[result] (\gate{2}{2}) -- (\gatephi{2}{2});
    \draw[result] (\gatephi{2}{2}) -- (phiin22);
    \draw[result] (\gatephi{2}{2}) -- (\gate{3}{1});
    \draw[result] (\gate{3}{1}) -- (\gatephi{3}{1});
    \draw[result] (\gatephi{3}{1}) -- (phiin31);
    \draw[result] (\gate{2}{2}) -- (\gate{2}{3});
    \draw[result] (\gatephi{2}{2}) -- (\gatephi{2}{3});
    \draw[result] (\gate{3}{1}) -- (\gate{3}{2});
    \draw[result] (\gate{3}{2}) -- (\gate{3}{3});
    \draw[result] (\gatephi{3}{1}) -- (\gatephi{3}{2});
    \draw[result] (\gatephi{3}{2}) -- (\gatephi{3}{3});

    \draw[result] (\gatephi{4}{2}) -- (\gate{5}{1}) -- (\gatephi{5}{1});
    \draw[result] (\gatephi{4}{2}) -- (\gatephi{4}{3});
    \draw[result] (\gate{5}{1}) -- (\gate{5}{2}) -- (\gate{5}{3});
    \draw[result] (\gatephi{5}{1}) -- (\gatephi{5}{2}) -- (\gatephi{5}{3});
    \draw[result] (\gatephi{4}{2}) -- (phiin42);
    \draw[result] (\gatephi{5}{1}) -- (phiin51);
    
  \end{tikzpicture}}
  \caption{The circuit for the bounded formula $\phi \buntil{3} \psi$ from
  \prettyref{fig:non-planar}.  If the $\psi_i$-gates evaluate to the constants shown in the bottom line, then the 
  circuit depicted as a gray-colored overlay is an equivalent planar circuit.}
  \label{fig:planar}
\end{center}
\end{figure}
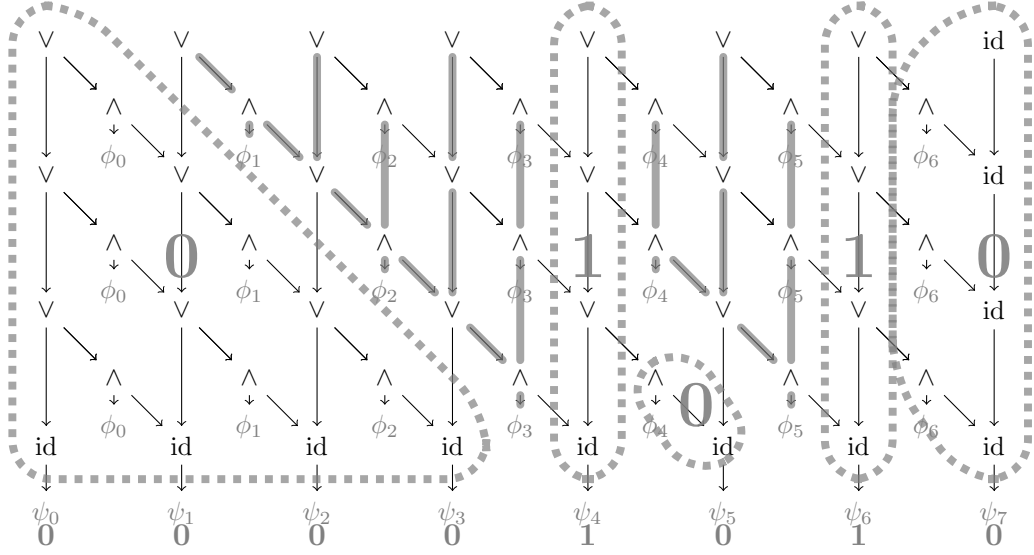%

%% file: figure-contraction.tex
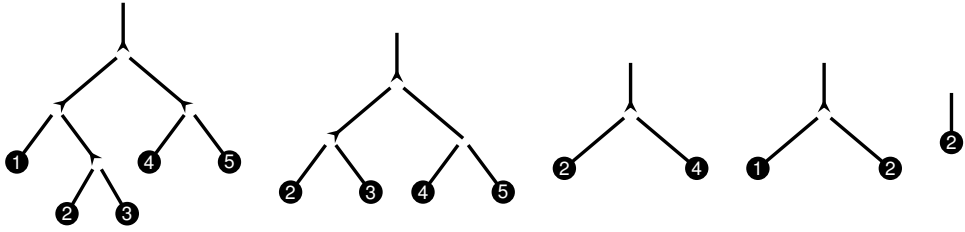
\begin{figure}
\tikzstyle{cc}=[>=latex,inner sep=0pt,outer sep=0pt]

\tikzstyle{cirtree}=[ultra thick,color=black,-stealth reversed,shorten <=4pt,line cap=rounded]
\tikzstyle{cirleaf}=[-]
\tikzstyle{cirleafn}=[inner sep=.5pt,outer sep=0pt,draw,circle,font=\sf\footnotesize,text=white,fill]
\tikzstyle{ctree}=[]


\scalebox{0.8}{%
\begin{tikzpicture}[baseline=0,level/.style={level distance=1cm-#1,sibling distance=4cm/#1},level 2/.style={sibling distance=2.2cm}]
  \coordinate (r) [cirtree]
    child {coordinate (n)
    child {coordinate (n0)
      child [cirleaf] {node[cirleafn] (n00) {1}}
        child {coordinate (n01)
          child [cirleaf,black] {node[cirleafn] (n010) {2}}
          child [cirleaf,black] {node[cirleafn] (n011) {3}}
        }
      }
      child {coordinate (n1)
        child [cirleaf] {node[cirleafn] (n10) {4}}
        child [cirleaf] {node[cirleafn] (n11) {5}}
      }
    };
\end{tikzpicture}\hspace{.5cm}
\begin{tikzpicture}[baseline=.5cm,level/.style={level distance=1cm-#1,sibling distance=4cm/#1},level 2/.style={sibling distance=2.2cm}]
  \coordinate (r) [cirtree]
    child {coordinate (n)
      child {coordinate (n01)
          child [cirleaf,black] {node[cirleafn] (n010) {2}}
          child [cirleaf,black] {node[cirleafn] (n011) {3}}
      }
      child [-] {coordinate (n1)
        child [cirleaf] {node[cirleafn] (n10) {4}}
        child [cirleaf] {node[cirleafn] (n11) {5}}
      }
    };
\end{tikzpicture}\hspace{.5cm}
%
%
%
\begin{tikzpicture}[baseline=1cm,level/.style={level distance=1cm-#1,sibling distance=4cm/#1},level 2/.style={sibling distance=2.2cm}]
   \coordinate (r) [cirtree]
     child {coordinate (n)
         child [cirleaf,black] {node[cirleafn] (n010) {2}}
         child [cirleaf] {node[cirleafn] (n10) {4}}
     };
\end{tikzpicture}\hspace{.5cm}
\begin{tikzpicture}[baseline=1cm,level/.style={level distance=1cm-#1,sibling distance=4cm/#1},level 2/.style={sibling distance=2.2cm}]
  \coordinate (r) [cirtree]
    child {coordinate (n)
        child [cirleaf,black] {node[cirleafn] (n010) {1}}
        child [cirleaf] {node[cirleafn] (n10) {2}}
    };
\end{tikzpicture}\hspace{.5cm}
\begin{tikzpicture}[baseline=1.5cm,level/.style={level distance=1cm-#1,sibling distance=4cm/#1},level 2/.style={sibling distance=2.2cm}]
  \coordinate (r) [cirtree]
  child [cirleaf] {node[cirleafn] (n10) {2}};
\end{tikzpicture}}
\caption{An parallel contraction process as produced by \prettyref{alg:contraction}.}
\label{fig:contraction}
\end{figure}


%% file: figure-buntil2.tex
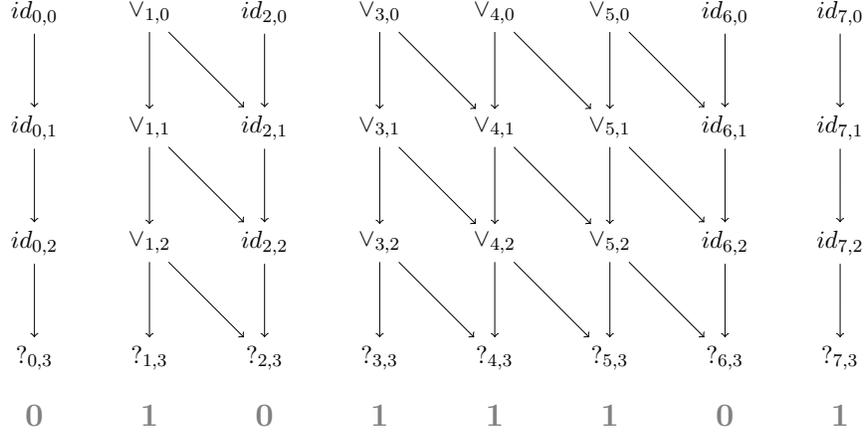
\begin{figure}[t]
  \begin{center}

    \newcommand{\mi}{7}
    \newcommand{\bound}{3}
    \newcommand{\gate}[2]{i{#1}j{#2}}
    \newcommand{\orgate}[2]{$\vee_{#1,#2}$}
    \newcommand{\andgate}[2]{$\wedge_{#1,#2}$}
    \newcommand{\idgate}[2]{$id_{#1,#2}$}
    \newcommand{\vargate}[2]{$?_{#1,#2}$}
    \def\param{0/0,1/1,2/0,3/1,4/1,5/1,6/0,7/1}
    \tikzstyle{sig}=[->]

    \scalebox{0.85}{%
    \begin{tikzpicture}[scale=1.8]
    
    \foreach \i/\s in \param {
    \foreach \j in {0,...,\bound} {
    \ifthenelse { \j<\bound \AND \( \s=0 \OR \i=\mi \) }
    { \node (\gate{\i}{\j}) at (\i,-\j) {\idgate{\i}{\j}}; }
    { \ifthenelse { \j<\bound }
    { \node (\gate{\i}{\j}) at (\i,-\j) {\orgate{\i}{\j}}; }
    { 
      \node (\gate{\i}{\j}) at (\i,-\j) {\vargate{\i}{\j}};
      \node[gray] at ( $ (\i, -\j)-(0,.5) $ ) {\Large\bf \s};
    }
    }
    }}

    \foreach \i/\s in \param {
    \foreach \j in {0,...,\bound} {
    \pgfmathtruncatemacro{\ni}{(\i+1)}
    \pgfmathtruncatemacro{\nj}{(\j+1)}
    \ifthenelse { \j<\bound \AND \( \s=0 \OR \i=\mi \) }
    { \draw [sig] (\gate{\i}{\j}) -- (\gate{\i}{\nj}); }
    {\ifthenelse { \j<\bound }
    {\draw [sig] (\gate{\i}{\j}) -- (\gate{\i}{\nj});
     \draw [sig] (\gate{\i}{\j}) -- (\gate{\ni}{\nj}); }
    {}
    }
    }}

  \end{tikzpicture}}
  \caption{%
  The circuit $\circuit{\rho}(s,\buntil{3})$ for $\abs{\rho} = 8$. The bottom
  line shows an example evaluation $s=0,1,0,1,1,1,0,1$ of the left subformula.
  }
  \label{fig:buntil2}
\end{center}
\end{figure}

%% file: figure-buntil1.tex
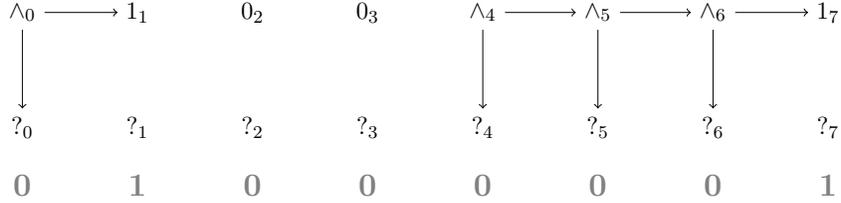
\begin{figure}[t]
  \begin{center}

    \newcommand{\mi}{7}

    \newcommand{\gate}[1]{i{#1}}
    \newcommand{\igate}[1]{ii{#1}}

    \newcommand{\onegate}[1]{$1_{#1}$}
    \newcommand{\nullgate}[1]{$0_{#1}$}
    \newcommand{\orgate}[1]{$\vee_{#1}$}
    \newcommand{\andgate}[1]{$\wedge_{#1}$}
    \newcommand{\idgate}[1]{$id_{#1}$}
    \newcommand{\vargate}[1]{$?_{#1}$}

    \scalebox{0.85}{%
    \begin{tikzpicture}[scale=1.8]
   
      \def\param{0/0/and, 1/1/one, 2/0/null, 3/0/null, 4/0/and, 5/0/and, 6/0/and, 7/1/one}

    \foreach \i/\s/\g in \param {
      \node (\gate{\i}) at (\i,0) {\csname\g gate\endcsname{\i}};
      \node (\igate{\i}) at (\i,-1) {\vargate{\i}};
      \node[gray] at (\i,-1.5) {\Large\bf \s};
    }

    \foreach \i/\s/\g in \param { 
    \def\x{and}
    \ifnum \i<\mi \ifx\x\g
        \pgfmathtruncatemacro{\ni}{(\i+1)}
        \draw[->] (\gate{\i}) -- (\gate{\ni});
        \draw[->] (\gate{\i}) -- (\igate{\i});
      \fi \fi
    }

  \end{tikzpicture}}
  \caption{%
  The circuit $\circuit{\rho}(\buntil{3},s)$ for $\abs{\rho} = 8$. The bottom
  line shows an example evaluation $s=0,1,0,0,0,0,0,1$ of the right subformula.
  }
  \label{fig:buntil1}
\end{center}
\end{figure}